\documentclass[a4paper,UKenglish,cleveref, autoref, thm-restate]{lipics-v2021}

\usepackage{mathtools}
\usepackage{tikz}
\usepackage{tipa}
\usetikzlibrary{calc}
\usetikzlibrary{positioning}

\hideLIPIcs  

\title{Even shorter proofs without new variables} 
\author{Adri\'an Rebola-Pardo}{Vienna University of Technology, Austria \and Johannes Kepler Universit\"at Linz, Austria}{adrian.rebola_pardo@jku.at}{https://orcid.org/0000-0001-9234-4377}{}
\authorrunning{A. Rebola-Pardo}
\Copyright{Adri\'an Rebola-Pardo}
\keywords{Interference, SAT solving, Unsatisfiability proofs, Unsatisfiable cores}
\ccsdesc{Hardware~Theorem proving and SAT solving}
\funding{This work has been supported by the LIT AI Lab funded by the State of Upper Austria,
the LogiCS doctoral program W1255-N23 of the Austrian Science Fund (FWF),
the Vienna Science and Technology Fund (WWTF) through grants VRG11-005 and ICT15-103, and
Microsoft Research through its PhD Scholarship Programme.}
\acknowledgements{I would like to thank Georg Weissenbacher for the discussions on WSR,
as well as the anonymous reviewers who provided very useful comments.}
\nolinenumbers

\newcommand{\restr}[2]{{\left.\kern-\nulldelimiterspace#1\vphantom{\big|}\right|_{#2}}}
\newcommand{\Set}[1]{\left\{#1\right\}}
\newcommand{\Cla}[1]{\left[#1\right]}
\newcommand{\Cub}[1]{\left\langle#1\right\rangle}
\newcommand{\smodels}{\mathrel{\models_{\textrm{sat}}}}
\newcommand{\sequiv}{\mathrel{\equiv_{\textrm{sat}}}}
\newcommand{\comp}{\overline}
\newcommand{\mut}[2]{\nabla{#1}.\,{#2}}

\begin{document}

\maketitle

\begin{abstract}
Proof formats for SAT solvers have diversified over the last decade, enabling new features such as extended resolution-like capabilities,
very general extension-free rules, inclusion of proof hints, and pseudo-boolean reasoning. Interference-based methods
have been proven effective, and some theoretical work has been undertaken to better explain their limits and semantics. In this work,
we combine the subsumption redundancy notion from~\cite{BussT19} and the overwrite logic framework from~\cite{Rebola-PardoS18}.
Natural generalizations then become apparent, enabling even shorter proofs of the pigeonhole principle (compared to those from~\cite{HeuleKB17})
and smaller unsatisfiable core generation.
\end{abstract}

\section{Introduction}
\label{sec:intro}

The impressive recent improvements in SAT solving have come coupled with the need to ascertain their results.
While satisfiability results are straightforward to check, unsatisfiability results require massive
proofs, sometimes petabytes in size~\cite{HeuleKM16,Heule18}.
The search for proof systems that enable both easy proof generation and smaller proofs has yield many
achievements~\cite{GoldbergN03,Gelder12,WetzlerHH14,HeuleKB17,Rebola-PardoC18,AltmanningerP20,BussT19,GochtN21,BaekCH21}.

Modern proof systems rely on redundancy properties presenting a phenomenon
known as \emph{interference}~\cite{JarvisaloHB12, HeuleK17B, Rebola-PardoS18}.
Whereas traditional proof systems derive clauses that are implied by the premises,
interference-based proof systems merely require introduced clauses to be consistent with them.
Interference proofs preserve the existence of a model throughout the proof, rather than models themselves.
A somewhat counterintuitive semantics thus arises: introducing a clause
in an interference-based proof system does not only depend on the presence of some clauses,
but also on the absence of some other clauses~\cite{PhilippR17,Rebola-PardoS18}.

The most general interference-based proof system in the literature is known as DSR~\cite{BussT19}.
While its predecesor DPR had success in generating short proofs of the pigeonhole formula
without introducing new variables~\cite{HeuleKB17}, DSR did not seem to succeeded in improving this result,
despite being intuitively well-suited for it.

In this work, we analyze the semantics of DSR proofs extending previous work on DPR proofs~\cite{Rebola-PardoS18}.
We find similar results to that article; in particular, satisfiability-preserving DSR proofs
can be reinterpreted as more traditional, DAG-shaped, model-preserving proofs over
an extension of propositional logic with a \emph{mutation} operator.
Crucially, these DAG-shaped proofs remove the whole-formula dependence interference is characterized by,
enabling an easier analysis of the necessary conditions for
satisfiability-preservation.

This analysis hints at a generalization we call
\emph{weak substition redundancy} (WSR \textipa{[\textprimstress w\textsci z\textschwa\textrhoticity]}),
which allows shorter, more understandable, easier to generate, faster to check proofs.
We demonstrate this by giving an even shorter proof of the pigeonhole formula.
We also provide a couple of examples where smaller unsatisfiable cores can be generated during proof checking,
and fewer lemmas are required during proof generation.

\paragraph*{Interference-based proofs}

Much of proof generation and checking is still done in the same way as a couple decades ago, by
logging the sequence of \emph{learnt clauses} in CDCL checkers, sometimes together with antecedents, and checking those
clauses for simple entailment criteria such as \emph{reverse unit propagation}~(RUP)~\cite{GoldbergN03,ZhangM03}.
Other parts of the proof are generated using more advanced deduction techniques;
even their infrequent use can dramatically decrease the size of
generated proofs~\cite{HeuleHW13a,WetzlerHH14,KieslRH18,HeuleB18},
overcoming not only technical limitations in proof generation,
but also theoretical bounds~\cite{Haken85,Urquhart87,Urquhart99}.
Clause deletion information is also recorded in the proof, which is needed to reduce memory
footprint in checking~\cite{HeuleHW14}.

Much research has been invested on finding ever more powerful proof rules~\cite{JarvisaloHB12,HeuleKB17,BussT19}
that allow to succintly express inprocessing techniques such as
Gaussian elimination~\cite{SoosNC09,Soos10,PhilippR16,ChewH20,GochtN21} or
symmetry breaking~\cite{AloulRMS03,AloulSM06,HeuleHW15}.
These proof rules are collectively called \emph{interference-based rules},
since their derivation depends on the whole formula
rather than just on the presence of some specific clauses~\cite{JarvisaloHB12,HeuleK17B,PhilippR17,Rebola-PardoS18}.
One of the most general interference techniques is \emph{substitution redundancy} (SR), which allows a version of
reasoning without loss of generality~\cite{BussT19}; this technique has been recently lifted to
pseudo-Boolean reasoning with impressive results~\cite{GochtN21}.

\paragraph*{Substitution redundancy and the pigeonhole problem}

A previous version of SR, called \emph{propagation redundancy}~(PR)~\cite{HeuleKB17}, was successful
in achieving short proofs of the pigeonhole problem, known for having exponential proofs in resolution~\cite{Haken85}
and polynomial yet cumbersome proofs in extended resolution~\cite{Cook76}.
The proof from~\cite{HeuleKB17} can be understood in terms of reasoning without loss of
generality~\cite{Rebola-PardoS18}: it assumes that a given pigeon is in a given pigeonhole,
for otherwise we could swap pigeons around.

PR does not have a method to swap the values of variables;
rather, it can only conditionally set them to true or false.
Hence, linearly many reasoning steps are needed to just to achieve the swap.
SR, on the other hand, allows variable swaps, so one could expect that the clause expressing the result
of this swap would satisfy the SR property. Surprisingly, it does not;
in fact, the clause fails to satisfy a requirement that
in its PR version was almost trivial.

\paragraph*{Interference and logical dependency}

Interference-based proofs do not have a ``dependence'' or ``procedence'' structure:
since the ability to introduce a clause is contingent on the whole formula,
no notion of ``antecedents'' exists for SR and its predecessors.
This becomes a problem when computing unsatisfiable cores and trimmed proofs~\cite{NadelRS13};
it also has the potential
to harm the performance of proof checkers, since some techniques that allow skipping unnecessary steps
during proof checking are based on logical dependence~\cite{HeuleHW13}.

This also relates to an issue arising when generating proof fragments for inprocessing techniques.
Sometimes, a clause $C$ cannot be introduced as SR because some lemmas are needed;
the proof generator might know these lemmas and how to derive them.
However, because interference depends on the whole formula,
introducing the lemmas before $C$ can further constrain the requirements
for $C$ to be introduced, demanding yet more lemmas.

\paragraph*{Contributions}

Previous work showed that the semantics of PR can be expressed in terms of \emph{overwrite logic}~\cite{Rebola-PardoS18}.
Overwrite logic extends propositional logic with an \emph{overwrite operator}.
Within overwrite logic, DPR proofs can be regarded as DAG-shaped, model-preserving proofs;
PR introduction can then be shown to behave as reasoning without loss of generality.

In Section~\ref{sec:mutation} we provide an extension to the overwrite logic framework, called \emph{mutation logic},
which elucidates the semantics of DSR proofs. In particular,
model-preserving proofs within mutation logic mimicking satisfiability-preserving DSR proofs
can be extracted, as shown in Section~\ref{ssc:entailment}. This allows a clearer understanding of the SR redundancy rule,
which in turn makes some improvements over SR apparent.

By introducing minor modifications to the definition of SR, in Section~\ref{sec:extensions} we obtain a new,
more powerful redundancy rule called \emph{weak substitution redundancy}~(WSR).
WSR proofs are more succint than DSR proofs, which we demonstrate
by providing a shorter proof of the pigeonhole problem using only $O(n^2)$ clause introductions in Section~\ref{ssc:php}.

Furthermore, WSR enables finer-grained ways to reason about dependency in interference-based proofs.
This can yield shorter proof checking runtimes and smaller trimmed proofs and unsatisfiability cores
when SR clauses are used (Section~\ref{ssc:cores}), as well as easier proof generation techniques by providing
clearer separation for interference lemmas (Section~\ref{ssc:lemmas}).

\section{Preliminaries}
\label{sec:prelim}

Given a \emph{literal} $l$, we denote its \emph{complement} as $\overline l$.
We denote \emph{clauses} by juxtaposing its literals within square brackets,
i.e.\ we denote the clause $l_1 \vee l_2 \vee l_3$ as $\Cla{l_1 l_2 l_3}$.
We similarly denote conjunctions of literals, called \emph{cubes}, as juxtaposed
literals within angle brackets, e.g.\ $\Cub{l_1 l_2 l_3}$.
Crucially, we only consider clauses and cubes that do not contain complementary
literals, as most SAT solvers and proof checkers already make that assumption.
Equivalently, we disallow tautological clauses and unsatisfiable cubes.
We also define complementation for clauses and cubes, i.e.
$\comp{\Cla{l_1 \dots l_n}} = \Cub{\comp{l_1}\dots\comp{l_n}}$ and
$\comp{\Cub{l_1 \dots l_n}} = \Cla{\comp{l_1}\dots\comp{l_n}}$.
SAT solving typically operates over formulas in \emph{conjunctive normal form} (CNF),
which are conjunctions of clauses. Here we regard CNF formulas as finite sets of clauses.

An \emph{atom} is either a literal, or one of the symbols $\top$ or $\bot$ representing
the propositional constants true and false. Complementation is extended to atoms with
$\overline{\top} = \bot$ and $\overline{\bot} = \top$. We can then define the usual propositional
semantics as follows. A \emph{model} $I$ is a total map
from atoms to $\Set{\top, \bot}$ such that $I(\top) = \top$ and $\overline{I(l)} = I(\overline l)$ for all atoms $l$.

We say that $I$ \emph{satisfies} a literal $l$ (written $I \models l$) whenever
$I(l) = \top$. This definition is recursively extended in the usual way to clauses
(disjunctively), cubes and CNF formulas (conjunctively). Similarly, we use the
typical notions of \emph{entailment} (denoted ${\models}$), logical \emph{equivalence}
($\equiv$) and \emph{satisfiability}. We also say that a logical expression $\varphi$
\emph{satisfiability-entails} another expression $\psi$ (denoted $\varphi \smodels \psi$)
whenever, if $\varphi$ is satisfiable, then $\psi$ is satisfiable too. Similarly, $\varphi$
is \emph{satisfiability-equivalent} to $\psi$ (denoted $\varphi \sequiv \psi$) whenever
$\varphi$ and $\psi$ are either both satisfiable or both unsatisfiable (i.e.\ $\varphi \smodels \psi$ and $\psi \smodels \varphi$).

An \emph{atomic substitution} $\sigma$ is a total map from atoms to atoms
satisfying the following constraints:
\begin{romanenumerate}
    \item $\sigma(\top) = \top$.
    \item $\overline{\sigma(l)}$ = $\sigma(\overline l)$ for all atoms $l$.
    \item $\sigma(l) \neq l$ only for finitely many atoms $l$.
\end{romanenumerate}
This definition is essentially equivalent to the substitutions from \cite{BussT19}.
The form presented here makes it easier to compose atomic substitutions with other atomic
substitutions, i.e.\ $(\sigma \circ \tau)(l) = \sigma(\tau(l))$,
and with models, i.e.\ $(I \circ \sigma)(l) = I(\sigma(l))$; the latter is a model
that satisfies a given logical expression $\varphi$ iff $I$ satisfies the expression
resulting from applying the substitution $\sigma$ to $\varphi$.

Note that atomic substitutions have a finite representation: only finitely
many literals are mapped to atoms other than themselves, and
giving the mapping for one polarity fixes the mapping for the other polarity.
Hence, one can represent a substitution as a set of mappings
$\Set{x_1 \mapsto l_1,\dots,x_n \mapsto l_n}$ where the $x_i$ are pairwise distinct
variables, the $l_i$ are atoms, and any variable other than the $x_i$ is mapped to itself.

Our restriction that clauses must be non-tautological is somewhat at odds
with the concept of substitutions. An atomic substitution $\sigma$ \emph{trivializes}
a clause $C$ if either:
\begin{alphaenumerate}
    \item there is a literal $l \in C$ with $\sigma(l) = \top$.
    \item there are two literals $l, k \in C$ with $\sigma(l) = \overline{\sigma(k)}$.
\end{alphaenumerate}
Applying $\sigma$ to $C$ yields a tautology whenever $\sigma$ trivializes $C$, and a
(non-tautological) clause otherwise.
Then we can define the \emph{reduct} of a clause $C$ or a CNF formula $F$ by an atomic
substitution $\sigma$ as:
\begin{align*}
    \restr{C}{\sigma} = {} & \Cla{\sigma(l) \mid l \in C \text{ and } \sigma(l) \neq \bot} \text{, \quad if $\sigma$ does not trivialize $C$} \\
    \restr{F}{\sigma} = {} & \Set{\restr{C}{\sigma} \mid C \in F \text{ and $\sigma$ does not trivialize C}}
\end{align*}

\begin{lemma}
\label{lem:reduct}
    Let $C$ be a clause, $F$ be a CNF formula, and $\sigma$ be an atomic substitution.
    The following then hold:
    \begin{romanenumerate}
        \item \label{itm:reduct:trivial}
        $\sigma$ trivializes $C$ if and only if $I \circ \sigma \models C$
        for all models $I$.
        \item \label{itm:reduct:clause}
        If $\sigma$ does not trivialize $C$, then
        $I \circ \sigma \models C$ if and only if $I \models \restr{C}{\sigma}$
        for all models $I$.
        \item \label{itm:reduct:cnf}
        $I \circ \sigma \models F$ if and only if
        $I \models \restr{F}{\sigma}$ for all models $I$.
    \end{romanenumerate}
\end{lemma}
\begin{proof}
    Let us first show~(\ref{itm:reduct:clause}). First, observe that $I$ satisfies
    $\restr{C}{\sigma}$ if and only if $I$ satisfies $\sigma(l)$ for some literal $l \in C$.
    But this is equivalent to $(I \circ \sigma)(l) = \top$
    for some $l \in C$, which is precisely $I \circ \sigma \models C$.

    We now show~(\ref{itm:reduct:trivial}). The ``only if'' implication is straightforward
    from the definition of a trivializing substitution. For the ``if'' implication,
    we show that if $\sigma$ does not trivialize $C$, then $I \circ \sigma$ falsifies $C$
    for some model $I$. Claim~(\ref{itm:reduct:clause}) gives out that any model $I$ falsifying
    $\restr{C}{\sigma}$, which exists because it is a (non-tautological) clause, has this property.

    Claim~(\ref{itm:reduct:cnf}) then follows easily from claims~(\ref{itm:reduct:trivial})
    and~(\ref{itm:reduct:clause}).
\end{proof}

Note that, for atomic substitutions that only map variables to the constants $\top$ or $\bot$,
there exists a correspondence with cubes. In particular, given variables $x_1,\dots,x_n,y_1,\dots,y_m$,
the cube $Q$ is bijectively associated to the atomic substitution $Q^\star$ where:
\begin{align*}
    Q = {} & \Cub{x_1 \dots x_n \, \comp{y_1} \dots \comp{y_m}} &
    Q^\star = {} & \Set{x_1 \mapsto \top, \dots, x_n \mapsto \top, y_1 \mapsto \bot, \dots, y_m \mapsto \bot}
\end{align*}

\subsection{Interference-based redundancy notions}
\label{ssc:interference}

Throughout the last decade, several redundancy notions collectively called
\emph{interference-based rules} have appeared in the literature~\cite{JarvisaloHB12,HeuleKB17,HeuleK17B,BussT19}.
Originating from clause elimination techniques~\cite{JarvisaloBH10,JarvisaloBH12,KieslSTB18}, interference can be also used to introduce
clauses in the formula; unlike more classical techniques, though, these clauses
do not need to be implied by the formula, but rather \emph{consistent} with it.
Specifically, given a CNF formula $F$, introducing a clause $C$ through interference
requires that $F \sequiv F \cup \Set{C}$.

Many interference-based rules are based on a criterion for entailment called
\emph{reverse unit propagation} (RUP)~\cite{GoldbergN03}. A clause $C$ is called a \emph{RUP clause} over
a CNF formula $F$ whenever unit propagation applied to $F$ using the assumption literals
$\comp C$ yields a conflict; under these circumstances, it can be shown that $F \models C$.

RUP clauses can be characterized in terms of resolution proofs. In particular,
a clause $C$ is a RUP clause over $F$ if and only if $C$ can be derived from $F$ through a
derivation of a particular form, called a \emph{subsumption-merge chain}~\cite{PhilippR17}.
These are derivations as shown in Figure~\ref{fig:smchain}, starting with a subsumption inference
and followed by a number of \emph{resolution merges}, also known as \emph{self-subsuming resolutions}~\cite{EenB05}.
The specifics of subsumption-merge chains in relation to RUPs are not quite relevant for
our discussion; we direct the interested reader to~\cite{PhilippR17,Rebola-PardoW20}.
For us, it suffices to know that checking whether $C$ is a RUP clause over $F$ is essentially
the same as finding the subsumption-merge chain that derives $C$ from $F$~\cite{ZhangM03}.

\begin{figure}
\caption{General form of a subsumption-merge chain~\cite{PhilippR17,Rebola-PardoW20} deriving the clause $A_n$ from premises $E_0,\dots,E_n$.
$\textsc{sub}$ represents the subsumption rule, so it requires $E_0 \subseteq A_0$. $\textsc{res}$ represents the resolution rule,
which can be applied if there is a literal $l_i \in A_{i-1}$ with $\comp l_i \in E_i$; in this case, $A_i = A_{i-1} \setminus \Set{l_i} \vee E_i \setminus \Set{\vphantom{l_i}\smash{\comp l_i}}$.
Subsumption-merge chains additionally require that the $\textsc{res}$ inferences are actually self-subsuming~\cite{EenB05},
i.e.\ $E_i \setminus \Set{\vphantom{l_i}\smash{\comp l_i}} \subseteq A_{i-1}$. Under these conditions, the clause $A_n$ is a RUP clause over any CNF formula containing
$E_0,\dots,E_n$. Conversely, any RUP clause over $F$ can be derived as $A_n$ through a subsumption-merge chain from some clauses $E_0,\dots,E_n \in F$.
In fact, the $E_i$ are the reason clauses used during unit propagation in a RUP check in reverse ordering (up to a topologically-compatible reordering)~\cite{Gelder12,PhilippR17}.}
\label{fig:smchain}
\centering
\begin{tikzpicture}
    \node[inner sep = 0pt, outer sep = 0pt, anchor = base] (an) {$A_n$};
    \coordinate[above = 2.5ex of an.base] (mn) {};
    \coordinate[left = 1cm of mn] (ln) {};
    \coordinate[right = 1cm of mn] (rn) {};
    \draw[] (ln) -- (rn);
    \node[inner sep = 0pt, outer sep = 0pt, above right = 1ex and 0.5ex of ln, anchor = base west] (ap) {$A_{n-1}$};
    \node[inner sep = 0pt, outer sep = 0pt, above left = 1ex and 0.5ex of rn, anchor = base east] (en) {$E_n$};
    \node[inner sep = 0pt, outer sep = 0pt, above left = 3ex and 4ex of ap.base, anchor = base] (dots) {$\ddots$};
    \node[inner sep = 0pt, outer sep = 0pt, above left = 0ex and 2ex of dots, anchor = base] (a2) {$A_2$};
    \coordinate[above = 2.5ex of a2.base] (m2) {};
    \coordinate[left = 1cm of m2] (l2) {};
    \coordinate[right = 1cm of m2] (r2) {};
    \draw[] (l2) -- (r2);
    \node[inner sep = 0pt, outer sep = 0pt, above right = 1ex and 0.5ex of l2, anchor = base west] (a1) {$A_1$};
    \node[inner sep = 0pt, outer sep = 0pt, above left = 1ex and 0.5ex of r2, anchor = base east] (e2) {$E_2$};
    \coordinate[above = 2.5ex of a1.base] (m1) {};
    \coordinate[left = 1cm of m1] (l1) {};
    \coordinate[right = 1cm of m1] (r1) {};
    \draw[] (l1) -- (r1);
    \node[inner sep = 0pt, outer sep = 0pt, above right = 1ex and 0.5ex of l1, anchor = base west] (a0) {$A_0$};
    \node[inner sep = 0pt, outer sep = 0pt, above left = 1ex and 0.5ex of r1, anchor = base east] (en) {$E_1$};
    \coordinate[above = 2.5ex of a0.base] (m0) {};
    \coordinate[left = 0.5cm of m0] (l0) {};
    \coordinate[right = 0.5cm of m0] (r0) {};
    \draw[] (l0) -- (r0);
    \node[inner sep = 0pt, outer sep = 0pt, above = 1ex of m0, anchor = base] (e0) {$E_0$};
    \node[inner sep = 0pt, outer sep = 0pt, left = 1ex of ln, anchor = mid east, font = \footnotesize] {$\textsc{res}$};
    \node[inner sep = 0pt, outer sep = 0pt, left = 1ex of l2, anchor = mid east, font = \footnotesize] {$\textsc{res}$};
    \node[inner sep = 0pt, outer sep = 0pt, left = 1ex of l1, anchor = mid east, font = \footnotesize] {$\textsc{res}$};
    \node[inner sep = 0pt, outer sep = 0pt, left = 1ex of l0, anchor = mid east, font = \footnotesize] {$\textsc{sub}$};
\end{tikzpicture}
\end{figure}

Building on RUP clauses, many redundancy notions can be defined. The most relevant for our discussion
are, in increasing generality order, \emph{resolution-asymmetric tautologies} (RATs),
\emph{propagation redundancies} (PRs) and \emph{substitution-redundancies} (SRs):

\begin{definition}
\label{def:redundancy}
    Let $C$ be a clause and $F$ be a CNF formula.
    \begin{romanenumerate}
        \item We say $C$ is a RAT clause~\cite{JarvisaloHB12} over $F$ upon a literal $l$
        whenever $l \in C$ and, for every clause $D \in F$ with $\comp l \in D$, the
        expression $C \vee D \setminus \Set{\vphantom{l}\smash{\comp l}}$ is either a tautology or a RUP
        clause over $F$.
        \item We say $C$ is a PR clause~\cite{HeuleKB17} over $F$ upon a cube $Q$ whenever
        $Q \models C$ (i.e.\ $Q \cap C \neq \emptyset$) and each clause in $\restr{F}{Q^\star}$ is a RUP clause over
        $\restr{F}{{\comp C}^\star}$.
        \item We say $C$ is a SR clause~\cite{BussT19} over $F$ upon an atomic substitution
        $\sigma$ whenever $\sigma$ trivializes $C$ and each clause in $\restr{F}{\sigma}$
        is a RUP clause over $\restr{F}{{\comp C}^\star}$.
    \end{romanenumerate}
\end{definition}

For a given \emph{witness} (i.e.\ the literal $l$, the cube $Q$ or the substitution $\sigma$),
checking whether a clause $C$ is a RAT/PR/SR clause over $F$ upon the corresponding witness
is polynomial over the size of $F$. In particular, this check takes at most one RUP check for each
clause~\cite{BussT19}; and RUP checking is quadratic on the size of $F$~\cite{Gelder12}.
Finding the right witness is nevertheless NP-complete~\cite{HeuleKB17}.
These redundancy notions satisfy the general condition for interference:

\begin{theorem}
    Let $C$ be a clause and $F$ be a CNF formula, and assume either of the following:
    \begin{alphaenumerate}
        \item $C$ is a RAT clause over $F$ upon a literal $l$ \cite{JarvisaloHB12}.
        \item $C$ is a PR clause over $F$ upon some cube $Q$ \cite{HeuleKB17}.
        \item $C$ is an SR clause over $F$ upon some atomic substitution $\sigma$ \cite{BussT19}.
    \end{alphaenumerate}
    Then, $F \sequiv F \cup \Set{C}$.
\end{theorem}

In this paper we will mostly focus on substitution redundancy, which is the most general of them.
However, we will use an equivalent definition, as per~\cite[Lemma 5]{BussT19}: instead of the
condition that each clause in $\restr{F}{\sigma}$ is a RUP clause over $\restr{F}{{\comp C}^\star}$,
we require that, for each clause $D \in F$, either $\sigma$ trivializes $D$,
or $\comp C \models \restr{D}{\sigma}$, or the clause $C \vee \restr{D}{\sigma}$ is a RUP clause
over $F$.

\subsection{Proof systems for SAT solving}

RUP clauses provided the first effective solution to the problem of certifying an unsatisfiability
result from a SAT solver. In particular, learnt clauses in a CDCL SAT solver~\cite{SilvaS96} are
RUP clauses~\cite{GoldbergN03, Gelder12}, so checking that each clause in the list of learnt clauses
is a RUP clause over the previously derived formula amounts to certifying that the last
clause in the list is entailed by the solved formula. If that clause is the empty clause,
the list constitutes a refutation.

However, the proof complexity of RUP proofs is rather poor:
there exist many simple problems whose refutations in resolution-based proof systems, such as RUP,
are exponential on the size of the refuted formula~\cite{Haken85,Urquhart87,Urquhart99}.
In fact, this problem extends to (purely) CDCL SAT solvers, on which these results impose
a performance upper bound~\cite{PipatsrisawatD11,BeameKS04}.

To alleviate the impact of these results, some inprocessing techniques were developed,
including reencoding of cardinality constraints~\cite{BiereBLM14,MantheyHB12}, Gaussian
elimination over $\mathbb{Z}_2$~\cite{SoosNC09,Soos10} and symmetry
breaking~\cite{AloulRMS03,AloulSM06}. Unfortunately, the aforementioned limitations
still apply to the generated refutation, so emitting a RUP proof would still
take exponential time.

Allowing interference-based reasoning in the proof led
to a vast number of proof formats~\cite{HeuleHW13a,WetzlerHH14,HeuleKB17,Cruz-FilipeMS17,Cruz-FilipeHHKS17,Lammich17,TanHM21,BussT19,BaekCH21}
and proof generation techniques~\cite{SinzB06,MantheyP14,HeuleHW15,PhilippR16,ChewH20,BryantH21,GochtN21,BryantBH22}.
The proof complexity of these systems is equivalent to that of
extended resolution~\cite{Tseitin1983,reckhow75_phd,KieslRH18,HeuleB18}, for which no
exponential lower bounds are known.

Unlike more traditional, DAG-shaped proofs, interference-based proofs take the
form of a list of \emph{clause introductions} and \emph{deletions}.
Starting with the input CNF formula $F$, clause introductions of the form $\textbf{i: }C$
add a clause $C$ to $F$, whereas clause deletions of the form $\textbf{d: }C$ remove
$C$ from $F$. At each point in the proof there is an \emph{accumulated formula}
where all the previous instructions in the proof have been applied.

Just as DAG-shaped proofs like resolution maintain a soundness invariant (i.e.\ each
model satisfying the premises of the proof also satisfies the conclusion),
interference-based proofs are \emph{satisfiability-preserving}~\cite{PhilippR17}:
at any point in an interference-based proof of $F$, the accumulated formula $G$
satisfies $F \smodels G$. This is guaranteed by imposing some conditions on
clause introductions; clause deletions do not have any requirements, because
deleting a clause is always satisfiability-preserving.

Different proof systems then arise from different conditions on clause introductions.
\emph{Delete Resolution Asymmetric Tautology} (DRAT) requires them to be either RUP clauses
or RAT clauses over the accumulated formula~\cite{HeuleHW13a, WetzlerHH14},
and similarly for \emph{Delete Propagation Redundancy}~(DPR)~\cite{HeuleKB17}
and \emph{Delete Substitution Redundancy}~(DSR)~\cite{BussT19}.
Note that, in the case of introducing a RAT/PR/SR clause (as opposed to a RUP clause),
the witness $\omega$ must be specified; in this case we denote it as $\textbf{i: }C,\,\omega$.

\subsection{Overwrite logic}
\label{ssc:overwrite}

Interference-based proofs represent a structural and semantic departure from traditional
proof systems. This is due to the \emph{non-monotonic} properties of SR: an SR clause
over $F$ upon~$\sigma$ is not necessarily an SR clause over a formula containing $F$.~\cite{JarvisaloHB12,PhilippR17}.

The consequences of non-monotonicity are far-reaching.
Interference-based proofs cannot be freely composed as, for example,
resolution proofs can~\cite{HeuleB15b}: the correctness of a clause introduction depends,
in principle, on the whole formula, which motivated the name ``interference'' as opposed
to ``inference''~\cite{HeuleK17B}.

DPR proofs can be seen as model-preserving, tree-shaped, monotonic proofs over a more general logic, known
as overwrite logic~\cite{Rebola-PardoS18}. There, a model $I$ can be \emph{conditionally overwritten}
with an \emph{overwrite rule} of the form $(Q \coloneq T)$, where $Q$ and $T$ are cubes.
Then, the model $I \circ (Q \coloneq T)$ is defined as $I \circ Q^\star$ if $I \models T$,
or as $I$ otherwise. That is, if $T$ is satisfied, then the minimal assignment satisfying $Q$
is overwritten on $I$.
Instead of clauses, overwrite logic deals with
\emph{overwrite clauses}, represented as $\mut{\varepsilon_1 \dots \varepsilon_n}{C}$,
where $C$ is a clause and the $\varepsilon_i = (Q_i \coloneq T_i)$ are overwrite rules.
Such an overwrite clause is satisfied by a model $I$ whenever
$I \circ \varepsilon_1 \circ \dots \circ \varepsilon_n \models C$.

This framework accurately expresses the reasoning performed by PR
introduction~\cite{Rebola-PardoS18}:
\begin{theorem}
    Let $C$ be a PR clause over a CNF formula $F$ upon a cube $Q$. Then, the implication
    $F \models \mut{(Q \coloneq \comp C)}{(F \cup \Set{C})}$ holds.
\end{theorem}
This result means that non-monotonic, satisfiability-preserving reasoning using PR clauses
can be turned into monotonic, model-preserving reasoning in overwrite logic.
\cite{Rebola-PardoS18} further introduces a traditional, DAG-shaped proof system over
overwrite clauses that mimics PR proofs, hence suggesting that the whole-formula dependence
featured by interference-based proof systems can, to some extent, be curbed.

\section{Mutation semantics for DSR proofs}
\label{sec:mutation}

The overwrite logic presented in Section~\ref{ssc:overwrite} was designed to formalize
the semantics of DPR proofs. In particular, models are overwritten with cubes, which act as
witnesses for PR clause introductions.
In order to extend this framework to DSR proofs, the role of cubes must now be fulfilled
by atomic substitutions. Here we introduce \emph{mutation logic},
which is a straightforward extension of overwrite logic.

In its most general form, a \emph{mutation rule} is an expression $(\sigma \coloneq \tau)$,
where $\sigma$ is an atomic substitution and $\tau$ is any logical expression that can be
evaluated under a model. We call $\tau$ the \emph{trigger} of the rule, and $\sigma$ its
\emph{effect}. Mutation rules themselves are not logical expressions and they cannot be
satisfied or falsified. They are instead intended to codify the idea ``if the trigger
$\tau$ is satisfied, then apply the effect substitution $\sigma$''.
We thus define the \emph{application} of a mutation rule $(\sigma \coloneq \tau)$ to a model
$I$ as:
\begin{equation*}
    I \circ (\sigma \coloneq \tau) = {}
    \begin{cases}
        I \circ \sigma & \text{if }I \models \tau \\
        I & \text{if }I \not\models \tau
    \end{cases}
\end{equation*}

As with overwrite logic, the main difference with propositional logic is the inclusion
of a mutation operator $\nabla$. As in~\cite{Rebola-PardoS18}, one can recursively
define mutation formulas as either propositional formulas, or expressions of the
form $\mut{(\sigma \coloneq \tau)}{\varphi}$ where $\sigma$ is an atomic substitution
and $\varphi$, $\tau$ are mutation formulas. The semantics of the mutation operator are
given by $I \models \mut{(\sigma \coloneq \tau)}{\varphi}$ whenever
$I \circ (\sigma \coloneq \tau) \models \varphi$. In other words: evaluating
$\mut{(\sigma \coloneq \tau)}{\varphi}$ corresponds to evaluating a formula $\varphi^\prime$
obtained from $\varphi$ by applying the effect $\sigma$ to $\varphi$
only if the trigger $\tau$ is satisfied.

This framework is very general, but just as discussed in~\cite{Rebola-PardoS18},
nothing meaningful is lost by introducing some strong restrictions. For the purpose of this paper,
we will only consider \emph{cubic} mutation rules of the form $(\sigma \coloneq Q)$ where $Q$ is a
propositional cube. The logical expressions we will use are of three kinds, where we use
$\mut{\vec{\varepsilon}}{\varphi}$ to denote a nested mutation
$\mut{\varepsilon_1}{{} \dots \mut{\varepsilon_n} \varphi}$ with cubic mutations $\varepsilon_i$:
\begin{itemize}
\item \emph{Mutation clauses} of the form $\mut{\vec{\varepsilon}}{C}$
where $C$ is a propositional clause.
\item \emph{Mutation CNF formulas} (MCNF), which are finite sets of mutation clauses.
The semantics of MCNF formulas are conjunctive, i.e.\ they are satisfied if every mutation
clause in them is satisfied.
\item \emph{Uniformly mutation CNF formulas} (UMCNF) of the form $\mut{\vec{\varepsilon}}{F}$
where $F$ is a propositional CNF formula. $\nabla$ distributes over the propositional
connectives, e.g.\ 
$\mut{\vec{\varepsilon}}{(\varphi_1 \wedge \varphi_2)} \equiv (\mut{\vec{\varepsilon}}{\varphi_1}) \wedge (\mut{\vec{\varepsilon}}{\varphi_2})$.
Hence, UMCNF can be embedded in the fragment of the MCNF formulas that contain clauses
with the same mutation prefix.
\end{itemize}

Similarly to how overwrite logic allows the expression of PR clauses as model-preserving
inferences under an overwrite~\cite{Rebola-PardoS18}, SR clauses become consequences
under a mutation.

\begin{theorem}
\label{thm:srsemantics}
    Let $F$ be a CNF formula and $C$ be an SR clause over $F$ upon an atomic
    mutation $\sigma$. Then, $F \models \mut{(\sigma \coloneq \comp C)}{(F \cup \Set{C})}$.
\end{theorem}
\begin{proof}
    Let $I$ be any model with $I \models F$. Our goal is to show that the model
    $I^\prime = I \circ (\sigma \coloneq \comp C)$ satisfies $F \cup \Set{C}$.
    If $I \models C$ holds, then $I^\prime = I$, which satisfies both $F$ and $C$.
    
    Let us now show the case with $I \not\models C$, where we have $I^\prime = I \circ \sigma$.
    First observe that, since $C$ is an SR clause upon $\sigma$, the clause $C$ is
    trivialized by $\sigma$. Lemma~\ref{lem:reduct} then shows $I^\prime \models C$.
    Now, consider any clause $D \in F$. By the definition of SR clauses, either $\sigma$
    trivializes $D$, or $\comp C \models \restr{D}{\sigma}$,
    or the clause $C \vee \restr{D}{\sigma}$ is a RUP clause over $F$.

    As above, the first case implies $I^\prime \models D$.
    For the second and third cases, it suffices to show $I \models \restr{D}{\sigma}$,
    since Lemma~\ref{lem:reduct} then proves $I^\prime \models D$.
    For the second case, this follows from $I \models \comp C$.
    For the third case, it follows from $I \models F$ and $I \not\models C$.
    We have thus shown that $I^\prime \models F \cup \Set{C}$ as we wanted.
\end{proof}

As for PR clauses in~\cite{Rebola-PardoS18}, one can read Theorem~\ref{thm:srsemantics} as
claiming that SR clause introduction (and in general, interference-based reasoning)
performs reasoning \emph{without loss of generality}. In particular:
$C$ can be assumed in $F$ because, were it not to hold in a given model of $F$,
a transformation, namely the one given by $\sigma$, could be applied to the variables
such that $F$ is still satisfied after the transformation, and $C$ becomes satisfied too.

\subsection{DSR proofs as model-preserving proofs}
\label{ssc:entailment}

The entailment in Theorem~\ref{thm:srsemantics} raises the question whether SR proofs
can be equivalently expressed as model-preserving, DAG-shaped proofs over the corresponding
mutated clauses. Following~\cite{Rebola-PardoS18}, we can define a proof system as shown in
Figure~\ref{fig:inferences}.

\begin{figure}
\caption{A proof system over mutation clauses.}
\label{fig:inferences}
\centering
\begin{tabular}{r@{}l}
\begin{tikzpicture}[baseline = (resx.mid east)]
    \node[inner sep = 0pt, outer sep = 0pt, font = \footnotesize] (resx) {$\textsc{res }$};
\end{tikzpicture}&\begin{tikzpicture}[baseline = (resl)]
    \node[inner sep = 0pt, outer sep = 0pt, anchor = base] (res) {$\mut{\vec{\varepsilon}}{C \setminus \Set{l} \vee D \setminus \Set{\comp l}}$};
    \coordinate[above left = 2.5ex and 0.5ex of res.base west] (resl) {};
    \coordinate[above right = 2.5ex and 0.5ex of res.base east] (resr) {};
    \draw[] (resl) -- (resr);
    \node[inner sep = 0pt, outer sep = 0pt, above right = 1.3ex and 0.5ex of resl, anchor = base west] (resc) {$\mut{\vec{\varepsilon}}{C}$};
    \node[inner sep = 0pt, outer sep = 0pt, above left = 1.3ex and 0.5ex of resr, anchor = base east] (resd) {$\mut{\vec{\varepsilon}}{D}$};
\end{tikzpicture}\\[4ex]
\begin{tikzpicture}[baseline = (subx.mid east)]
    \node[inner sep = 0pt, outer sep = 0pt, font = \footnotesize] (subx) {$\textsc{sub }$};
\end{tikzpicture}&\begin{tikzpicture}[baseline = (subl)]
    \node[inner sep = 0pt, outer sep = 0pt, anchor = base] (subd) {$\mut{\vec{\varepsilon}}{D}$};
    \coordinate[above left = 2.5ex and 0.5ex of subd.base west] (subl) {};
    \coordinate[above right = 2.5ex and 0.5ex of subd.base east] (subr) {};
    \draw[] (subl) -- (subr);
    \node[inner sep = 0pt, outer sep = 0pt, above = 3.8ex of subd.base, anchor = base] (subc) {$\mut{\vec{\varepsilon}}{C}$};
    \node[inner sep = 0pt, outer sep = 0pt, right = 2ex of subr, anchor = mid west, font = \small] {where $C \subseteq D$};
\end{tikzpicture}\\[4ex]
\begin{tikzpicture}[baseline = (tautx.mid east)]
    \node[inner sep = 0pt, outer sep = 0pt, font = \footnotesize] (tautx) {$\nabla\textsc{taut }$};
\end{tikzpicture}&\begin{tikzpicture}[baseline = (tautl)]
    \node[inner sep = 0pt, outer sep = 0pt, anchor = base] (taut) {$\mut{\vec{\varepsilon}}{\mut{(\sigma \coloneq \comp C)}{C}}$};
    \coordinate[above left = 2.5ex and 0.5ex of taut.base west] (tautl) {};
    \coordinate[above right = 2.5ex and 0.5ex of taut.base east] (tautr) {};
    \draw[] (tautl) -- (tautr);
    \node[inner sep = 0pt, outer sep = 0pt, right = 2ex of tautr, anchor = mid west, font = \small] {where $\sigma$ trivializes $C$};
\end{tikzpicture}\\[4ex]
\begin{tikzpicture}[baseline = (intx.mid east)]
    \node[inner sep = 0pt, outer sep = 0pt, font = \footnotesize] (intx) {$\nabla\textsc{intro }$};
\end{tikzpicture}&\begin{tikzpicture}[baseline = (intl)]
    \node[inner sep = 0pt, outer sep = 0pt, anchor = base] (int) {$\mut{\vec{\varepsilon}}{\mut{(\sigma \coloneq Q)}{C}}$};
    \coordinate[above left = 2.5ex and 3.5ex of int.base west] (intl) {};
    \coordinate[above right = 2.5ex and 3.5ex of int.base east] (intr) {};
    \draw[] (intl) -- (intr);
    \node[inner sep = 0pt, outer sep = 0pt, above right = 1.3ex and 0.5ex of intl, anchor = base west] (intc) {$\mut{\vec{\varepsilon}}{C}$};
    \node[inner sep = 0pt, outer sep = 0pt, above left = 1.3ex and 0.5ex of intr, anchor = base east] (intd) {${}^{\displaystyle\star}\mut{\vec{\varepsilon}}{\comp Q \vee \restr{C}{\sigma}}$};
    \node[inner sep = 0pt, outer sep = 0pt, right = 2ex of intr, anchor = mid west, font = \small] {where $\star$ is only needed if $Q \not\models \restr{C}{\sigma}$};
\end{tikzpicture}\\[4ex]
\begin{tikzpicture}[baseline = (elimx.mid east)]
    \node[inner sep = 0pt, outer sep = 0pt, font = \footnotesize] (elimx) {$\nabla\textsc{elim }$};
\end{tikzpicture}&\begin{tikzpicture}[baseline = (eliml)]
    \node[inner sep = 0pt, outer sep = 0pt, anchor = base] (elim) {$\mut{\vec{\varepsilon}}{\restr{C}{\sigma}}$};
    \coordinate[above left = 2.5ex and 18.5ex of elim.base west] (eliml) {};
    \coordinate[above right = 2.5ex and 18.5ex of elim.base east] (elimr) {};
    \draw[] (eliml) -- (elimr);
    \node[inner sep = 0pt, outer sep = 0pt, above right = 1.3ex and 0.5ex of eliml, anchor = base west] (elimc) {$\mut{\vec{\varepsilon}}{\mut{(\sigma \coloneq Q)}{C}}$};
    \node[inner sep = 0pt, outer sep = 0pt, above left = 1.3ex and 0.5ex of elimr, anchor = base east] (elimd) {$\mut{\vec{\varepsilon}}{\mut{(\sigma \coloneq Q)}{\restr{C}{\sigma}}}$};
    \node[inner sep = 0pt, outer sep = 0pt, right = 2ex of elimr, anchor = mid west, font = \small] {where $\sigma$ does not trivialize $C$};
\end{tikzpicture}
\end{tabular}
\end{figure}

\begin{theorem}
    The inference rules in Figure~\ref{fig:inferences} are sound, i.e.\ any model satisfying the premises of each
    rule satisfies its conclusion as well.
\end{theorem}
\begin{proof}
    The proofs for $\textsc{res}$ and $\textsc{sub}$ are straightforward, since the $\nabla$ operator
    preserves implications.

    For $\nabla\textsc{taut}$, consider any model $I$, and let $I^\prime = I \circ \vec{\varepsilon} \circ (\sigma \coloneq \comp C)$.
    If $I \circ \vec{\varepsilon} \models C$, then $I^\prime = I \circ \vec{\varepsilon}$, which satisfies $C$.
    Otherwise, $I^\prime = I \circ \vec{\varepsilon} \circ \sigma$, and since $\sigma$ trivializes $C$ we have $I^\prime \models C$.

    Let us now show $\nabla\textsc{elim}$ correct. Consider any model $I$ satisfying the premises, and call
    $I^\prime = I \circ \vec{\varepsilon} \circ (\sigma \coloneq \comp C)$, so that $I^\prime \models C$ and $I^\prime \models \restr{C}{\sigma}$.
    If $I \circ \vec{\varepsilon} \models Q$, then $I^\prime = I \circ \vec{\varepsilon} \circ \sigma$ satisfies $C$;
    then $I \circ \vec{\varepsilon}$ satisfies $\restr{C}{\sigma}$ by Lemma~\ref{lem:reduct}.
    Otherwise, $I^\prime = I \circ \vec{\varepsilon}$ satisfies $\restr{C}{\sigma}$.

    Finally, for $\nabla\textsc{intro}$, let $I$ be any model satisfying the premises, and $I^\prime = I \circ \vec{\varepsilon} \circ (\sigma \coloneq Q)$.
    If $I \circ \vec{\varepsilon}$ satisfies $Q$ then either it also satisfies $\comp Q \vee \restr{C}{\sigma}$ or $Q \models \restr{C}{\sigma}$.
    Either way, we can conclude $I \circ \vec{\varepsilon} \models \restr{C}{\sigma}$, and since in this case we have $I^\prime = I \circ \vec{\varepsilon} \circ \sigma$,
    Lemma~\ref{lem:reduct} implies that $I^\prime \models C$.
    The other case is $I \circ \vec{\varepsilon} \not\models Q$, and in this case $I^\prime = I \circ \vec{\varepsilon}$, which satisfies $C$.
\end{proof}

Upon closer inspection of the proof of Theorem~\ref{thm:srsemantics},
the relation between the SR property and satisfiability-preservation becomes clearer.
When each clause $D \in F$ is required that either $\sigma$ trivializes $D$, or
$\comp C \models \restr{D}{\sigma}$, or the clause $C \vee \restr{D}{\sigma}$ is a RUP
clause over $F$, these conditions enable deriving
$\mut{(\sigma \coloneq \comp C)}{D}$ through rules $\nabla\textsc{taut}$
or $\nabla\textsc{intro}$ hold: the left-hand premise in $\nabla\textsc{intro}$
just means that $C$ has been derived earlier in the SR proof, while the
right-hand premise ensures that $C \vee \restr{D}{\sigma}$ can be derived (e.g.\ as a RUP clause).

On the other hand, $\nabla\textsc{taut}$ guarantees that $\mut{(\sigma \coloneq \comp C)}{C}$
can be derived, since the definition of SR clauses \emph{forces} $C$ to be trivialized
by $\sigma$. Given that $\nabla$ distributes over $\wedge$, these conditions are proving
$F \models \mut{(\sigma \coloneq \comp C)}{F}$.

Similar to~\cite{Rebola-PardoS18}, a translation of a DSR proof into a mutation logic proof then works as follows.
At each step in the DSR proof, we consider the list of rules $(\sigma_i \coloneq \comp C_i)$
corresponding to each SR clause $C_i$ introduced upon $\sigma_i$ earlier in the proof;
this list is potentially empty, e.g.\ at the start of the proof.
Let us denote this list by $\vec{\varepsilon}$. Then, at that point, all clauses
$D$ in the accumulated CNF formula have been derived as mutation clauses
$\mut{\vec{\varepsilon}}{D}$ in the translation.
The translation then proceeds as follows:
\begin{enumerate}
\item Deletions in the DSR proof are not translated.
\item A RUP clause $C$ can be derived through a subsumption-merge chain~\cite{PhilippR17}; rules $\textsc{res}$
and $\textsc{sub}$ can express a similar derivation of the mutation version of $C$.
\item For an SR clause $C$ over a CNF formula $F$ upon an atomic substitution $\sigma$,
we must derive mutation clauses $D_\nabla = \mut{\vec{\varepsilon}}{\mut{(\sigma \coloneq \comp C)}{D}}$ for
each clause $D \in F \cup \Set{C}$.
\begin{enumerate}
    \item When $\sigma$ trivializes $D$, the mutation clause $D_\nabla$ can be derived
    as an axiom through $\nabla\textsc{taut}$. Note that this case includes the case
    $D = C$ as well; this detail will become relevant in Section~\ref{ssc:wsr}.
    \item When $\comp C \models \restr{D}{\sigma}$, the mutation clause $D_\nabla$ can be
    infered from $\mut{\vec{\varepsilon}}{D}$ through $\nabla\textsc{intro}$. We know this premise has been
    previously derived because $D \in F$.
    \item When $C \vee \restr{D}{\sigma}$ is a RUP clause over $F$, a subsumption-merge chain
    deriving that clause with premises in $F$ exists~\cite{Gelder12, PhilippR17}.
    Replacing clauses $D^\prime$ with mutation clauses $\mut{\vec{\varepsilon}}{{D^\prime}}$,
    resolution inferences with $\nabla\textsc{res}$ and subsumption inferences with
    $\nabla\textsc{sub}$ in that proof then yields a derivation of
    $\mut{\vec{\varepsilon}}{C \vee \restr{D}{\sigma}}$ from previously derived mutation clauses.
    Finally, the rule $\nabla\textsc{intro}$ derives $D_\nabla$.
\end{enumerate}
\item At the end of the proof, the empty clause $\Cla{\,}$ is derived in the SR clause, and
the translation has derived the mutated clause $\mut{\varepsilon}{\Cla{\,}}$. The identity
$\Cla{\,} = \restr{\Cla{\,}}{\sigma}$ for all substitutions $\sigma$ ensures that
$\nabla\textsc{elim}$ can be iteratively applied to eliminate all mutation operators,
so that $\Cla{\,}$ is derived in the translation as well.
\end{enumerate}

\section{Extending DSR proofs}
\label{sec:extensions}

Understanding DSR proofs as mutation logic proofs opens the door to finer-grained
reasoning about interference-based proofs. Crucially, one of the main issues with
interference-based proofs is that deriving a clause involves reasoning over the whole
currently derived formula. In particular, interference-based proofs can be highly
\emph{non-monotonic}: deleting a clause in the current formula can enable new SR
introductions; and conversely, introducing a clause can disable previously available
SR introductions.

This is, at first sight, at odds with the translation described in Section~\ref{ssc:entailment}:
the proofs we obtain there are model-preserving, DAG-shaped proofs with clear
dependencies with other derived clauses. What can be derived in a subproof is never
affected by independent proof sub-DAGs, so clause introduction never disables SR
introductions. Deletions are even more intriguing, since they do not even exist in
the mutation logic framework (just as there is no notion of deletion in a resolution
proof DAG).

Another noticeable feature is how differently an SR clause $C$ over $F$ is treated in
the definition compared to the clauses $D \in F$. Even if at first sight it might
look reasonable to consider different conditions on the premises and on the conclusion,
the translation from Section~\ref{ssc:entailment} uses the same set of inference rules
to derive both $C_\nabla$ and $D_\nabla$.

\subsection{Weak substitution redundancy}
\label{ssc:wsr}

In the translation, the conditions of the definition are used to guarantee that
$C_\nabla$ can be derived through a $\nabla\textsc{taut}$ inference.
However, we have \emph{three} rules that can derive this
mutated clause, and the three are involved in deriving $D_\nabla$ for each $D \in F$.
We can thus relax the conditions over $C$ by demanding just the same as for each $D$:
either $\sigma$ must trivialize $C$, or $\comp C \models \restr{C}{\sigma}$,
or the clause $C \vee \restr{C}{\sigma}$ must be a RUP clause over $F$.

Furthermore, there is nothing in the translation forcing us to derive
$D_\nabla$ for \emph{each and all} clauses $D \in F$.
Rather, we must only do so for those clauses that the proof uses later on.
However, even if we do not need $D$ after the SR introduction, we still can 
use $D$ for the RUP checks of \emph{other} clauses in $F$. Note that this is not
quite the same as deleting $D$ before the SR introduction: doing so could make
the RUP checks of other clauses in $F$ fail.

These two details suggest an extension of substitution redundancy, which we call
\emph{weak substitution redundancy} (WSR).

\begin{definition}
A clause $C$ is a WSR clause over a CNF
formula $F$ upon an atomic substitution $\sigma$ modulo a subformula $\Delta \subseteq F$
whenever, for each clause $D \in (F \setminus \Delta) \cup \Set{C}$, either of the following holds:
\begin{alphaenumerate}
    \item $\sigma$ trivializes $D$.
    \item $\comp C \models \restr{D}{\sigma}$.
    \item $C \vee \restr{D}{\sigma}$ is a RUP clause over $F$.
\end{alphaenumerate}
\end{definition}

\begin{theorem}
Let $C$ be a WSR clause over a CNF formula $F$ upon an atomic substitution $\sigma$ modulo
a subformula $\Delta \subseteq F$. Then, $F \models \mut{(\sigma \coloneq \comp C)}{((F \setminus \Delta) \cup \Set{C})}$ holds.
In particular, if $F$ is satisfiable, then so is $(F \setminus \Delta) \cup \Set{C}$.
\end{theorem}
\begin{proof}
    Similar to the proof of Theorem~\ref{thm:srsemantics}. The main difference is that
    $F \models \mut{(\sigma \coloneq \comp C)}{C}$ must now be shown using the same reasoning
    as $F \models \mut{(\sigma \coloneq \comp C)}{D}$ for $D \in F$.
\end{proof}

The complexity of checking a WSR clause introduction is similar to that of a PR/SR check.
On the one hand, one extra RUP check might be needed if $C \vee \restr{C}{\sigma}$ is not
a tautology; on the other hand, one RUP check is spared for each clause in $\Delta$.

A minor benefit of WSR clauses is that, while not every RUP is a RAT, PR or SR clause,
every RUP clause is a WSR clause upon the identity atomic substitution.
The reason for this is that the condition that
the atomic substitution $\sigma$ must trivialize $\Cla{\,}$ always fails.
This allows reasoning about WSR proofs without the need for case discussion.
\begin{corollary}
    Let $C$ be clause and $F$ be a CNF formula. Then, $C$ is a RUP clause over $F$ if and only if
    $C$ is a WSR clause over $F$ upon the identity atomic substitution modulo $\emptyset$.
\end{corollary}

This, together with the embedded notion of deletions as $\Delta$, enables the definition of a proof
system with only one rule $\textbf{w: }C,\,\sigma \setminus \Delta$. This rule introduces clause $C$ and
deletes clauses in $\Delta$, and is correct whenever $C$ is a WSR clause over $F$ upon $\sigma$ modulo
$\Delta$. We call this proof system the \emph{WSR proof system}.

Note that the relaxations that derive the WSR proof system from DSR already exist in the literature,
albeit in different frameworks. In~\cite{GochtN21}, the first relaxation appears in a proof system
that derives pseudo-Boolean constraints rather than clauses. This work is extended later in~\cite{0001GMN22}
by explicitly maintaining two different accumulated \emph{core} and \emph{derived} formulas,
much in the vein of~\cite{JarvisaloHB12}. A \emph{transfer rule} allows moving core constraints to
the derived constraints, whereas a \emph{dominance-based strengthening rule} allows deriving redundant core constraints while only checking
SR-like conditions on other core constraints, hence treating derived constraints somewhat similarly to the ``modulo'' clauses in WSR.
However, the framework from~\cite{0001GMN22} targets optimization solving approaches; correspondingly,
dominance-based strengthening also includes side conditions on the objective function.
It is not immediately clear what these side conditions morph into when considered over a non-optimization framework,
nor whether the separation between core and derived constraints is fluid enough to fully simulate our second relaxation.

\section{Applications of WSR proofs}
\label{sec:applications}

So far, we have not yet shown any benefit of WSR over SR (or that they are not equivalent,
for that matter). In this section, we demonstrate techniques using WSR proofs that are unavailable
in previously existing interference-based proof systems.

\subsection{A shorter proof of the pigeonhole problem}
\label{ssc:php}

One of the first propositional problems that was found to only have exponential
resolution proofs was the pigeonhole problem~\cite{Haken85}. While polynomial proofs
in the extended resolution system had already been known for a decade~\cite{Cook76},
these proofs needed to introduce fresh variables to support definitions.
However, the seminal work on PR clauses presented a shorter DPR proof that did not
use extra variables, using $O(n^3)$ instructions~\cite{HeuleKB17}.

In~\cite{Rebola-PardoS18} an analysis of this proof from the overwrite logic perspective
was presented; let us briefly reproduce it here. The pigeonhole problem encodes the unsatisfiable
problem ``find an assignment of $n$ pigeons to $n-1$ pigeonholes such that no two pigeons share
the same hole''.
We consider variables $p_{ij}$ encoding ``pigeon $i$ is in hole $j$''. Let us define the following clauses:
\begin{align*}
    H_{in} = {} & \Cla{p_{ij} \mid 1 \leq j < n} \quad \text{for $n > 0$ and $1 \leq i \leq n$} \\
    P_{ijk} = {} & \Cla{\comp{p_{ik}}\,\comp{p_{jk}}} \quad \text{for $1 \leq i < j$ and $1 \leq k$} \\
    L_{ijn} = {} & \Cla{\comp{p_{i (n - 1)}}\,\comp{p_{nj}}} \quad \text{for $n > 1$, $1 \leq i < n$ and $1 \leq j < n - 1$} \\
    R_{in} = {} & \Cla{\comp{p_{i (n - 1)}}} \quad \text{for $n > 1$ and $i < n$}
\end{align*}
Briefly, $H_{in}$ says that pigeon $i$ stays in some hole $1 \leq j < n$; $P_{ijk}$ prevents that pigeons $i$
and $j$ both occupy hole $k$; $L_{ijn}$ can be read as ``if the pigeon $i$ is in the last hole, then hole $j$ does not contain the last pigeon'';
and finally $R_{in}$ prevents that pigeon $i$ is in the last hole.
The pigeonhole problem for $n$ pigeons is then encoded by
\begin{equation*}
\Pi_n = \Set{H_{in} \mid 1 \leq i \leq n} \cup \Set{P_{ijk} \mid 1 \leq i < j \leq n \text{ and }1 \leq k < n}
\end{equation*}
Intuitively, a refutation of $\Pi_n$ proceeds by noting that, without loss of generality, each pigeon $i < n$ is not in hole $n - 1$;
were this not the case, one can swap pigeon $i$ with pigeon $n$ (which is not in hole $n$ because that would violate $P_{in(n-1)}$).
Then, pigeons $1, \dots, (n - 1)$ and holes $1, \dots, (n - 2)$ are in the conditions of the pigeonhole problem $\Pi_{n - 1}$.
This process can be iterated until $\Pi_1$ is reached, which is trivially unsatisfiable.

The proof from \cite{HeuleKB17} follows this reasoning, but a single PR clause is not expressive enough to encode
swaps: the only mutations that it can handle are setting variables to true or false.
Thus, the proof first derives clauses $L_{ijn}$ for $1 \leq i < n$ and $1 \leq j < n - 1$ as PR clauses with
the cube $Q_{ijn} = \Cub{\comp{p_{i (n - 1)}}\,\comp{p_{nj}}p_{ij}p_{n (n - 1)}}$. This encodes the following reasoning:
without loss of generality, if the pigeon $i$ is in the last hole, then hole $j$ does not contain the last pigeon;
were this not the case, ensure that pigeon $i$ is not in the last hole but in the hole $j$ instead, and that
the last pigeon is not in hole $j$ but in the last hole instead. Once the clauses $L_{ijn}$ have been derived for
each $1 \leq i < n$, the clause $R_{in}$ ensuring that pigeon $i$ is not in the last hole can be derived as a RUP clause.

When considered together, the mutations $Q_{ijn}^\star$ for $1 \leq j < n - 1$ express the atomic substitution that swaps
pigeons $i$ and $n$, that is:
\begin{equation*}
    \sigma_{in} = \Set{p_{ij} \mapsto p_{nj}, p_{nj} \mapsto p_{ij} \mid 1 \leq j < n} \quad \text{for $1 \leq i < n$}
\end{equation*}
DSR can handle this kind of mutation. Let us write a DSR derivation of $\Pi_{n - 1}$ from $\Pi_{n}$
(where we are omitting some trailing deletions for simplicity):
\begin{equation*}
    (\textbf{i: }R_{1n},\,\sigma_{1n}),\dots,(\textbf{i: }R_{(n-1)n},\,\sigma_{(n-1)n}),(\textbf{i: }H_{1(n-1)}),\dots,(\textbf{i: }H_{(n-1)(n-1)})
\end{equation*}
Clauses $H_{i(n-1)}$ can be introduced as RUP clauses, since they result from resolution on $H_{in}$ and $R_{in}$.
Furthermore, one would hope for the $R_{in}$ clauses to be SR clauses over the preceding formula upon $\sigma_{in}$.
Let us check this. For each clause $D$ in the preceding formula $F$, we need to check that either of the following holds:
\begin{alphaenumerate}
    \item $D$ is trivialized by $\sigma_{in}$
    \item $\Cub{p_{i (n - 1)}} \models \restr{D}{\sigma_{in}}$
    \item the clause $D_\nabla = \Cla{\comp {p_{i (n - 1)}}} \vee \restr{D}{\sigma_{in}}$ is a RUP clause over $F$.
\end{alphaenumerate}    
Checking case by case one can see that the reduct $\restr{D}{\sigma_{in}}$ is always another clause in $F$,
so $D_\nabla$ is either a tautology or can be derived by subsumption from $F$ (which implies it is a RUP clause).

The clause $R_{in}$, nevertheless, is not an SR clause over $F$ upon $\sigma_{in}$, because it is not trivialized by $\sigma_{in}$.
Observe, however, that
\begin{equation*}
    C_\nabla = C \vee \restr{C}{\sigma_{in}} = \Cla{\comp {p_{i(n - 1)}}\,\comp {p_{n(n - 1)}}} = P_{in(n - 1)} \in F
\end{equation*}
In particular, $C_\nabla$ it is a RUP clause over $F$. Hence, $R_{in}$ is in fact a WSR clause over $F$ upon $\sigma_{in}$ modulo $\emptyset$.
Hence, we can define the WSR derivation $\pi_n$ of $\Pi_1$ from $\Pi_n$ given in Figure~\ref{fig:php} for $n > 1$.
The derivation $\pi_n$ has $O(n^2)$ instructions, and is in fact a refutation, since $\Cla{} \in \Pi_1$.

\begin{figure}
\caption{A WSR refutation $\pi_n$ of the pigeonhole problem $\Pi_n$ for $n \geq 1$ containing only $O(n^2)$ instructions.
Here, $\pi_1$ is the empty list and $\textrm{id}$ represents the identity atomic substitution.}
\label{fig:php}
\centering
\begin{tabular}{l}
    $\textbf{w: }R_{1n},\,\sigma_{1n}\setminus\emptyset$\\
    $\textbf{w: }R_{2n},\,\sigma_{2n}\setminus\emptyset$\\
    \qquad$\vdots$\\
    $\textbf{w: }R_{(n-1)n},\,\sigma_{(n-1)n}\setminus\emptyset$\\
    $\textbf{w: }H_{1(n-1)},\,\textrm{id}\setminus\Set{R_{1n}}$\\
    $\textbf{w: }H_{2(n-1)},\,\textrm{id}\setminus\Set{R_{2n}}$\\
    \qquad$\vdots$\\
    $\textbf{w: }H_{(n-1)(n-1)},\,\textrm{id}\setminus(\Set{R_{1n}} \cup (\Pi_n \setminus \Pi_{n - 1}))$\\
    $\pi_{n-1}$
\end{tabular}
\end{figure}

\subsection{Smaller cores and shorter checking runtime}
\label{ssc:cores}
SAT solvers generate proofs which often introduce clauses uninvolved in the derivation of a contradiction.
This is practically unavoidable because of how solvers generate proofs:
solvers mostly just log every learnt clause~\cite{GoldbergN03},
and at that point the solver does not know what learnt clauses will be useful.

State-of-the-art proof checkers thus validate the proof
backwards~\cite{HeuleHW13,WetzlerHH14}. Starting from the empty clause at the end of the proof,
the checker finds out what clauses are needed to derive each clause as a RUP clause.
Required clauses are then \emph{marked}; as the checker proceeds backwards, unmarked clauses
are skipped. If one were to visualize a RUP proof as a DAG, this amounts to only
checking the connected component that actually derives the empty clause while disregarding
all other connected components in the DAG.

Backwards checking has three interesting consequences. First, it vastly improves checking runtime:
not only are checks for unmarked clauses skipped, but also their premises are skipped as well
(unless they are used to derive another marked clause). Second, a shorter, \emph{trimmed} proof
can be extracted as a by-product of checking. Finally, by the time the checker reaches the start
of the proof, the marked clauses in the input formula form a (not necessarily minimal) unsatisfiable core.

\paragraph{Backwards checking in interference-based proofs}

Interference-based proofs do not have DAG-like dependencies as RUP proofs have.
Let us formalize the problem of backwards checking in this situation.
We assume that the checker keeps track of a CNF formula $F$ and
marked clauses $M \subseteq F$ as it proceeds backwards through the proof. When a
RAT/PR/SR introduction $\textbf{i: }C,\,\omega$ is reached with $C \in M$, the checker removes
$C$ from both the formula $F$ and the marked clauses $M$ and validates the
corresponding RAT/PR/SR introduction. The goal then is to find some (preferably small) subformula $M^\prime$
with $M \subseteq M^\prime \subseteq F$ such that $C$ is a RAT/PR/SR clause upon $\omega$ over $M^\prime$;
this will be the new set of marked clauses.

In the best case scenario, $C$ satisfies the corresponding redundancy property over $M$,
so the checker can move on with $M^\prime = M$.
There is only one way the redundancy property might not hold over $M$:
when one of the RUP checks from Definition~\ref{def:redundancy} fails over $M$
(but still succeeds over $F$), the premises of the induced subsumption-merge
chain must become marked; let us (conspicuously) call this set of \emph{newly} marked clauses $\Delta$.
The problem we are tackling is whether clauses in $\Delta$ \emph{really} need their own RUP check as
mandated by Definition~\ref{def:redundancy}.

For RAT, it turns out, they do not: one can show that, for a witness literal $l$ in a RAT check,
the clauses in $\Delta$ never contain $\comp l$, so they never trigger further RUP checks.
Such a convenient coincidence does not hold for PR or SR, though. In order to establish
that $C$ is a PR/SR clause upon some $M^\prime$, the clauses in $\Delta$ must undergo their own
RUP check, which might add new clauses to $\Delta$, and so on until fixpoint.

This is nevertheless wasteful. By the time the first $\Delta$ has been computed,
introducing $C$ can already be claimed to be satisfiability-preserving, \emph{just not as a PR/SR}:
the conditions above prove that $C$ is a WSR clause upon $\omega$ over $M \cup \Delta$ modulo $\Delta$.
This means that a proof checker (even one that only checks PR/SR) can simply set $M^\prime = M \cup \Delta$
and continue checking the rest of the proof.

To the best of our knowledge, existing checkers do not deal with this situation in
an optimal way, e.g.\ 
the reference DPR checker \texttt{dpr-trim} resorts instead to the fixpoint method%
\footnote{See \url{https://github.com/marijnheule/dpr-trim/blob/83eb40b9028100aca63a419eb6d08b45acf517ad/dpr-trim.c}, line 660.}.
Note that the fixpoint method always produces a larger $M^\prime$ than the WSR-based
method, with associated longer runtimes, larger unsatisfiability cores and longer
trimmed proofs.

Even if for (uncertified) checking WSR only seems relevant at a theoretical level,
state-of-the-art proof checkers emit trimmed, annotated proofs that can be further checked
with a verified tool~\cite{Cruz-FilipeHHKS17, TanHM21}. The formats these annotated proofs use,
such as LRAT or LPR, are based on RAT/PR, and so the fixpoint method is needed
if an annotated proof must be emitted in one of these formats.
Either way, the need for the fixpoint method could be removed by emitting WSR-based annotated proofs.

\begin{example}
    Let us define three CNF formulas. The formula $M$ contains clauses:
    \begin{align*}
        \Cla{a\,\comp c\,x} \qquad \Cla{\comp a\,\comp u\,\comp v\,\comp x} \qquad \Cla{c\,\comp u\,\comp v\,x} & \qquad \Cla{a\,\comp x\,\comp y\,\comp z} \qquad \Cla{a\,\comp c\,\comp x\,y} \qquad \Cla{\comp a\,b\,u}\\
        \Cla{c\,u} \qquad \Cla{\comp u\,y\,z} \qquad \Cla{\comp a\,\comp b\,\comp c} \qquad &\Cla{c\,\comp x\,\comp z}^{\bullet} \qquad \Cla{\comp c\,\comp x z}^{\bullet} \qquad \Cla{c\,\comp x\,\comp y}^{\bullet} \qquad \Cla{\comp a\,b\,\comp u\,v\,\comp x}^{\bullet}
    \end{align*}
    The formula $\Delta$ contains:
    \begin{equation*}
        \Cla{b\,\comp u\,x} \qquad \Cla{\comp b\,\comp t\,v\,x\,\comp y} \qquad \Cla{\comp b\,t\,v\,x\,\comp z} \qquad \Cla{t\,v\,\comp y\,z} \qquad \Cla{\comp t\,v\,y\,\comp z}
    \end{equation*}
    Finally, $\Gamma = \Set{\Cla{\comp b\,x\,\comp u\,y\,\comp z}}$. Let us assume that a proof checker is
    checking a DSR refutation of the unsafistiable formula $F = M \cup \Delta \cup \Gamma$ backwards. It eventually reaches the first instruction, an SR clause introduction
    for $C = \Cla{x\,\comp u}$ upon the atomic substitution $\sigma = \Set{x \mapsto \top, a \mapsto \top, v \mapsto \bot, t \mapsto \top}$. At this point, the clauses in $M$
    (in addition to $C$) have been marked for checking; since this is the first instruction, the marked clauses after checking $C$ for SR are an unsatisfiable core of $F$.
    One can check that $\sigma$ trivializes $C$, and that all clauses in $M$ except for the ones highlighted with~$\bullet$ satisfy the conditions in Definition~\ref{def:redundancy}
    using propagation clauses exclusively from $M$.
    For the highlighted clauses, propagating with clauses from $M \cup \Delta$ does suffice to satisfy Definition~\ref{def:redundancy}.

    As we learnt in Section~\ref{sec:extensions}, we can now stop checking: $C$ is a WSR clause over the formula $M \cup \Delta$ modulo $\Delta$;
    the newly marked clauses (which form the generated unsatisfiable core) are thus $M \cup \Delta$.
    Current checkers will nevertheless not stop here, since SR is more restrictive than WSR. In particular, they check the newly marked clauses $\Delta$
    for the conditions in Definition~\ref{def:redundancy} as well. As it turns out, $C$ is not even an SR clause over $M \cup \Delta$, but only over $F$:
    for the RUP check for $C \vee \restr{\Cla{\comp t\,v\,y\,\comp z}}{\sigma}$ to succeed, the clause in $\Gamma$ is needed too.
    That clause becomes subsequently marked, and a further check is performed for it. This check finally succeeds, reaching a fixpoint.

    This example shows that SR marks strictly more clauses than WSR, which translates into larger generated unsatisfiable cores and trimmed proofs, as well
    as a longer checking runtime since the extra marked clauses will be themselves checked.
\end{example}

\subsection{Interference-free interference lemmas}
\label{ssc:lemmas}

The differences between SR and WSR presented in Section~\ref{ssc:cores} can too be exploited
during proof generation. While the largest share of a proof generated by a state-of-the-art SAT solver
consists of learnt clauses introduced as RUP clauses as well as clause deletions, inprocessing techniques
also contribute to the proof. Typically, an inprocessing technique performs some reasoning
and then a (to some extent) hardcoded proof fragment of the results is generated.
No proof search is performed; rather, the specialized reasoning performed by the inprocessing technique
is translated into the target proof system by a method that has previously been proven correct (on paper, not \emph{in silico}).

Interference-based proof systems are notable for their ability to generate succint proof fragments
for many inprocessing techniques and non-CDCL methods, including parity reasoning~\cite{PhilippR16,GochtN21},
symmetry breaking~\cite{HeuleHW15} and BDD-based reasoning~\cite{BryantBH22}. Devising these proofs
is complex for several reasons; among them is that, in an interference-based proof system, introduced lemmas
may need further lemmas for satisfy Definition~\ref{def:redundancy}.

Let us assume we want to generate a proof fragment deriving a clause
$C$ as an SR clause from $F$ upon some atomic substitution $\sigma$.
The clause $C$ has been obtained through some inprocessing technique,
and we know that all the clauses $D_\nabla = C \vee \restr{D}{\sigma}$ for $D \in F$ are implied by $C$
because of some property of the inprocessing technique.
However, we might find that some of the $D_\nabla$ are not RUPs over $F$;
after all, RUP is just a criterion for entailment.
We can derive some additional clauses (i.e.\ lemmas) $L^1,\dots,L^n$ from $F$ such that
$D_\nabla$ is a RUP over $F \cup L^1,\dots,L^n$, but now the definition of SR clauses demands that
the $L^i_\nabla$ are RUP clauses as well, which might need additional clauses and so on.

This is, in essence, the proof generation version of the proof checking situation from Section~\ref{ssc:cores}.
Just as we did there, with WSR we can completely bypass the need to prove that the $L^i_\nabla$ are RUP clauses:
$C$ is already a WSR clause over $F \cup \Set{L^1,\dots,L^n}$ upon $\sigma$ modulo $\Set{L^1,\dots, L^n}$.
In other words, WSR allows introducing interference lemmas that need not be taken into account for RUP checks.

\begin{example}
\label{exp:lemmas}
Let us consider the CNF formula $F$ containing clauses:
\begin{align*}
    \Cla{a\,b\,\comp x\,y} \qquad \Cla{a\,b\,x\,y\,\comp z} \qquad \Cla{a\,b\,x\,z} \qquad &\Cla{\comp a\,u\,v} \qquad \Cla{\comp c\,u\,v} \qquad \Cla{a\,c\,\comp b\,y} \qquad \Cla{a\,c\,b\,\comp y}\\
    \Cla{c\,\comp b\,\comp y\,\comp z} \qquad \Cla{c\,\comp b\,x\,\comp y\,z} \qquad &\Cla{c\,\comp b\,\comp x\,z} \qquad \Cla{\comp u\,v} \qquad \Cla{u\,\comp v} \qquad \Cla{\comp u\,\comp v}
\end{align*}
We want to derive the clause $C = \Cla{x}$. Unfortunately, $C$ is not a RUP clause over $F$, so we try to introduce it
as an SR clause upon the atomic substitution $\sigma = \Set{x \mapsto \top, y \mapsto z, z \mapsto y}$. This \emph{almost} works:
all the conditions in Definition~\ref{def:redundancy} hold,
except for $C \vee \restr{\Cla{\comp b\,c\,\comp x\,z}}{\sigma} = \Cla{\comp b\,c\,x\,y}$ not being a RUP clause over $F$.
We can derive some RUP lemmas from $F$, for example $L_1 = \Cla{\comp a\,y\,v}$ and then $L_2 = \Cla{\comp a\,y}$;
the clause $\Cla{\comp b\,c\,x\,y}$ is indeed a RUP clause over $F \cup \Set{L_1, L_2}$.

Here is where WSR and SR show their differences again. Under WSR, we can already introduce $C$ in $F$,
because the paragraph above implies that $C$ is a WSR clause over $F \cup \Set{L_1, L_2}$ upon $\sigma$ modulo $\Set{L_1, L_2}$.
This is not the case for SR, though: because the clause $C \vee \restr{L_2}{\sigma} = \Cla{\comp a\,x\,z}$ is not a RUP clause
over $F \cup \Set{L_1, L_2}$, the clause $C$ is not SR over $F$ upon $\sigma$.
We would need to find additional lemmas to make it so, which might then need further lemmas themselves.
\end{example}

\section{Conclusion}
\label{sec:conclusion}

We have presented a generalization of the SR redundancy notion, called weak substitution redundancy (WSR).
This extension is straightforward once the semantics of interference have been understood,
which we achieve by extending the overwrite logic framework from~\cite{Rebola-PardoS18} into
mutation logic, which is able to handle atomic substitutions.

The main differences between SR and WSR are the weakening of one unnecessarily strong condition in the definition,
and the specification of a set of clauses that can be used for ensuring the interference conditions
but will not participate in interference themselves.

These minor differences have an impact on the versatility of the proof system. Shorter proofs can be obtained,
lemmas can be used in a less obstrusive way, the efficiency of the backwards checking algorithm is enhanced,
and smaller unsatisfiable cores and trimmed proofs can be generated.

\bibliography{main}

\begin{thebibliography}{10}

\bibitem{AloulRMS03}
Fadi~A. Aloul, Arathi Ramani, Igor~L. Markov, and Karem~A. Sakallah.
\newblock Solving difficult instances of boolean satisfiability in the presence
  of symmetry.
\newblock {\em {IEEE} Trans. Comput. Aided Des. Integr. Circuits Syst.},
  22(9):1117--1137, 2003.
\newblock \href {https://doi.org/10.1109/TCAD.2003.816218}
  {\path{doi:10.1109/TCAD.2003.816218}}.

\bibitem{AloulSM06}
Fadi~A. Aloul, Karem~A. Sakallah, and Igor~L. Markov.
\newblock Efficient symmetry breaking for boolean satisfiability.
\newblock {\em {IEEE} Trans. Computers}, 55(5):549--558, 2006.
\newblock \href {https://doi.org/10.1109/TC.2006.75}
  {\path{doi:10.1109/TC.2006.75}}.

\bibitem{AltmanningerP20}
Johannes Altmanninger and Adri{\'{a}}n Rebola{-}Pardo.
\newblock Frying the egg, roasting the chicken: unit deletions in {DRAT}
  proofs.
\newblock In Jasmin Blanchette and Catalin Hritcu, editors, {\em Proceedings of
  the 9th {ACM} {SIGPLAN} International Conference on Certified Programs and
  Proofs, {CPP} 2020, New Orleans, LA, USA, January 20-21, 2020}, pages 61--70.
  {ACM}, 2020.
\newblock \href {https://doi.org/10.1145/3372885.3373821}
  {\path{doi:10.1145/3372885.3373821}}.

\bibitem{BaekCH21}
Seulkee Baek, Mario~M. Carneiro, and Marijn J.~H. Heule.
\newblock A flexible proof format for {SAT} solver-elaborator communication.
\newblock In Jan~Friso Groote and Kim~Guldstrand Larsen, editors, {\em Tools
  and Algorithms for the Construction and Analysis of Systems - 27th
  International Conference, {TACAS} 2021, Held as Part of the European Joint
  Conferences on Theory and Practice of Software, {ETAPS} 2021, Luxembourg
  City, Luxembourg, March 27 - April 1, 2021, Proceedings, Part {I}}, volume
  12651 of {\em Lecture Notes in Computer Science}, pages 59--75. Springer,
  2021.
\newblock \href {https://doi.org/10.1007/978-3-030-72016-2\_4}
  {\path{doi:10.1007/978-3-030-72016-2\_4}}.

\bibitem{BeameKS04}
Paul Beame, Henry~A. Kautz, and Ashish Sabharwal.
\newblock Towards understanding and harnessing the potential of clause
  learning.
\newblock {\em J. Artif. Intell. Res.}, 22:319--351, 2004.
\newblock \href {https://doi.org/10.1613/jair.1410}
  {\path{doi:10.1613/jair.1410}}.

\bibitem{BiereBLM14}
Armin Biere, Daniel~Le Berre, Emmanuel Lonca, and Norbert Manthey.
\newblock Detecting cardinality constraints in {CNF}.
\newblock In Carsten Sinz and Uwe Egly, editors, {\em Theory and Applications
  of Satisfiability Testing - {SAT} 2014 - 17th International Conference, Held
  as Part of the Vienna Summer of Logic, {VSL} 2014, Vienna, Austria, July
  14-17, 2014. Proceedings}, volume 8561 of {\em Lecture Notes in Computer
  Science}, pages 285--301. Springer, 2014.
\newblock \href {https://doi.org/10.1007/978-3-319-09284-3\_22}
  {\path{doi:10.1007/978-3-319-09284-3\_22}}.

\bibitem{0001GMN22}
Bart Bogaerts, Stephan Gocht, Ciaran McCreesh, and Jakob Nordstr{\"{o}}m.
\newblock Certified symmetry and dominance breaking for combinatorial
  optimisation.
\newblock In {\em Thirty-Sixth {AAAI} Conference on Artificial Intelligence,
  {AAAI} 2022, Thirty-Fourth Conference on Innovative Applications of
  Artificial Intelligence, {IAAI} 2022, The Twelveth Symposium on Educational
  Advances in Artificial Intelligence, {EAAI} 2022 Virtual Event, February 22 -
  March 1, 2022}, pages 3698--3707. {AAAI} Press, 2022.
\newblock URL: \url{https://ojs.aaai.org/index.php/AAAI/article/view/20283}.

\bibitem{BryantBH22}
Randal~E. Bryant, Armin Biere, and Marijn J.~H. Heule.
\newblock Clausal proofs for pseudo-boolean reasoning.
\newblock In Dana Fisman and Grigore Rosu, editors, {\em Tools and Algorithms
  for the Construction and Analysis of Systems - 28th International Conference,
  {TACAS} 2022, Held as Part of the European Joint Conferences on Theory and
  Practice of Software, {ETAPS} 2022, Munich, Germany, April 2-7, 2022,
  Proceedings, Part {I}}, volume 13243 of {\em Lecture Notes in Computer
  Science}, pages 443--461. Springer, 2022.
\newblock \href {https://doi.org/10.1007/978-3-030-99524-9\_25}
  {\path{doi:10.1007/978-3-030-99524-9\_25}}.

\bibitem{BryantH21}
Randal~E. Bryant and Marijn J.~H. Heule.
\newblock Generating extended resolution proofs with a bdd-based {SAT} solver.
\newblock In Jan~Friso Groote and Kim~Guldstrand Larsen, editors, {\em Tools
  and Algorithms for the Construction and Analysis of Systems - 27th
  International Conference, {TACAS} 2021, Held as Part of the European Joint
  Conferences on Theory and Practice of Software, {ETAPS} 2021, Luxembourg
  City, Luxembourg, March 27 - April 1, 2021, Proceedings, Part {I}}, volume
  12651 of {\em Lecture Notes in Computer Science}, pages 76--93. Springer,
  2021.
\newblock \href {https://doi.org/10.1007/978-3-030-72016-2\_5}
  {\path{doi:10.1007/978-3-030-72016-2\_5}}.

\bibitem{BussT19}
Sam Buss and Neil Thapen.
\newblock {DRAT} proofs, propagation redundancy, and extended resolution.
\newblock In Mikol{\'{a}}s Janota and In{\^{e}}s Lynce, editors, {\em Theory
  and Applications of Satisfiability Testing - {SAT} 2019 - 22nd International
  Conference, {SAT} 2019, Lisbon, Portugal, July 9-12, 2019, Proceedings},
  volume 11628 of {\em Lecture Notes in Computer Science}, pages 71--89.
  Springer, 2019.
\newblock \href {https://doi.org/10.1007/978-3-030-24258-9\_5}
  {\path{doi:10.1007/978-3-030-24258-9\_5}}.

\bibitem{ChewH20}
Leroy Chew and Marijn J.~H. Heule.
\newblock Sorting parity encodings by reusing variables.
\newblock In Luca Pulina and Martina Seidl, editors, {\em Theory and
  Applications of Satisfiability Testing - {SAT} 2020 - 23rd International
  Conference, Alghero, Italy, July 3-10, 2020, Proceedings}, volume 12178 of
  {\em Lecture Notes in Computer Science}, pages 1--10. Springer, 2020.
\newblock \href {https://doi.org/10.1007/978-3-030-51825-7\_1}
  {\path{doi:10.1007/978-3-030-51825-7\_1}}.

\bibitem{Cook76}
Stephen~A. Cook.
\newblock A short proof of the pigeon hole principle using extended resolution.
\newblock {\em {SIGACT} News}, 8(4):28--32, 1976.
\newblock \href {https://doi.org/10.1145/1008335.1008338}
  {\path{doi:10.1145/1008335.1008338}}.

\bibitem{Cruz-FilipeHHKS17}
Lu{\'{\i}}s Cruz{-}Filipe, Marijn J.~H. Heule, Warren A.~Hunt Jr., Matt
  Kaufmann, and Peter Schneider{-}Kamp.
\newblock Efficient certified {RAT} verification.
\newblock In Leonardo de~Moura, editor, {\em Automated Deduction - {CADE} 26 -
  26th International Conference on Automated Deduction, Gothenburg, Sweden,
  August 6-11, 2017, Proceedings}, volume 10395 of {\em Lecture Notes in
  Computer Science}, pages 220--236. Springer, 2017.
\newblock \href {https://doi.org/10.1007/978-3-319-63046-5\_14}
  {\path{doi:10.1007/978-3-319-63046-5\_14}}.

\bibitem{Cruz-FilipeMS17}
Lu{\'{\i}}s Cruz{-}Filipe, Jo{\~{a}}o Marques{-}Silva, and Peter
  Schneider{-}Kamp.
\newblock Efficient certified resolution proof checking.
\newblock In Axel Legay and Tiziana Margaria, editors, {\em Tools and
  Algorithms for the Construction and Analysis of Systems - 23rd International
  Conference, {TACAS} 2017, Held as Part of the European Joint Conferences on
  Theory and Practice of Software, {ETAPS} 2017, Uppsala, Sweden, April 22-29,
  2017, Proceedings, Part {I}}, volume 10205 of {\em Lecture Notes in Computer
  Science}, pages 118--135, 2017.
\newblock \href {https://doi.org/10.1007/978-3-662-54577-5\_7}
  {\path{doi:10.1007/978-3-662-54577-5\_7}}.

\bibitem{EenB05}
Niklas E{\'{e}}n and Armin Biere.
\newblock Effective preprocessing in {SAT} through variable and clause
  elimination.
\newblock In Fahiem Bacchus and Toby Walsh, editors, {\em Theory and
  Applications of Satisfiability Testing, 8th International Conference, {SAT}
  2005, St. Andrews, UK, June 19-23, 2005, Proceedings}, volume 3569 of {\em
  Lecture Notes in Computer Science}, pages 61--75. Springer, 2005.
\newblock \href {https://doi.org/10.1007/11499107\_5}
  {\path{doi:10.1007/11499107\_5}}.

\bibitem{Gelder12}
Allen~Van Gelder.
\newblock Producing and verifying extremely large propositional refutations -
  have your cake and eat it too.
\newblock {\em Ann. Math. Artif. Intell.}, 65(4):329--372, 2012.
\newblock \href {https://doi.org/10.1007/s10472-012-9322-x}
  {\path{doi:10.1007/s10472-012-9322-x}}.

\bibitem{GochtN21}
Stephan Gocht and Jakob Nordstr{\"{o}}m.
\newblock Certifying parity reasoning efficiently using pseudo-boolean proofs.
\newblock In {\em Thirty-Fifth {AAAI} Conference on Artificial Intelligence,
  {AAAI} 2021, Thirty-Third Conference on Innovative Applications of Artificial
  Intelligence, {IAAI} 2021, The Eleventh Symposium on Educational Advances in
  Artificial Intelligence, {EAAI} 2021, Virtual Event, February 2-9, 2021},
  pages 3768--3777. {AAAI} Press, 2021.
\newblock URL: \url{https://ojs.aaai.org/index.php/AAAI/article/view/16494}.

\bibitem{GoldbergN03}
Evguenii~I. Goldberg and Yakov Novikov.
\newblock Verification of proofs of unsatisfiability for {CNF} formulas.
\newblock In {\em 2003 Design, Automation and Test in Europe Conference and
  Exposition {(DATE} 2003), 3-7 March 2003, Munich, Germany}, pages
  10886--10891. {IEEE} Computer Society, 2003.
\newblock URL:
  \url{http://doi.ieeecomputersociety.org/10.1109/DATE.2003.10008}, \href
  {https://doi.org/10.1109/DATE.2003.10008}
  {\path{doi:10.1109/DATE.2003.10008}}.

\bibitem{Haken85}
Armin Haken.
\newblock The intractability of resolution.
\newblock {\em Theor. Comput. Sci.}, 39:297--308, 1985.
\newblock \href {https://doi.org/10.1016/0304-3975(85)90144-6}
  {\path{doi:10.1016/0304-3975(85)90144-6}}.

\bibitem{HeuleHW13}
Marijn Heule, Warren A.~Hunt Jr., and Nathan Wetzler.
\newblock Trimming while checking clausal proofs.
\newblock In {\em Formal Methods in Computer-Aided Design, {FMCAD} 2013,
  Portland, OR, USA, October 20-23, 2013}, pages 181--188. {IEEE}, 2013.
\newblock URL: \url{http://ieeexplore.ieee.org/document/6679408/}.

\bibitem{HeuleHW13a}
Marijn Heule, Warren A.~Hunt Jr., and Nathan Wetzler.
\newblock Verifying refutations with extended resolution.
\newblock In Maria~Paola Bonacina, editor, {\em Automated Deduction - {CADE-24}
  - 24th International Conference on Automated Deduction, Lake Placid, NY, USA,
  June 9-14, 2013. Proceedings}, volume 7898 of {\em Lecture Notes in Computer
  Science}, pages 345--359. Springer, 2013.
\newblock \href {https://doi.org/10.1007/978-3-642-38574-2\_24}
  {\path{doi:10.1007/978-3-642-38574-2\_24}}.

\bibitem{HeuleHW14}
Marijn Heule, Warren A.~Hunt Jr., and Nathan Wetzler.
\newblock Bridging the gap between easy generation and efficient verification
  of unsatisfiability proofs.
\newblock {\em Softw. Test. Verification Reliab.}, 24(8):593--607, 2014.
\newblock \href {https://doi.org/10.1002/stvr.1549}
  {\path{doi:10.1002/stvr.1549}}.

\bibitem{HeuleHW15}
Marijn Heule, Warren A.~Hunt Jr., and Nathan Wetzler.
\newblock Expressing symmetry breaking in {DRAT} proofs.
\newblock In Amy~P. Felty and Aart Middeldorp, editors, {\em Automated
  Deduction - {CADE-25} - 25th International Conference on Automated Deduction,
  Berlin, Germany, August 1-7, 2015, Proceedings}, volume 9195 of {\em Lecture
  Notes in Computer Science}, pages 591--606. Springer, 2015.
\newblock \href {https://doi.org/10.1007/978-3-319-21401-6\_40}
  {\path{doi:10.1007/978-3-319-21401-6\_40}}.

\bibitem{HeuleK17B}
Marijn Heule and Benjamin Kiesl.
\newblock The potential of interference-based proof systems.
\newblock In Giles Reger and Dmitriy Traytel, editors, {\em {ARCADE} 2017, 1st
  International Workshop on Automated Reasoning: Challenges, Applications,
  Directions, Exemplary Achievements, Gothenburg, Sweden, 6th August 2017},
  EPiC Series in Computing, pages 51--54. EasyChair, 2017.
\newblock URL: \url{https://easychair.org/publications/paper/TWVW}.

\bibitem{Heule18}
Marijn J.~H. Heule.
\newblock Schur number five.
\newblock In Sheila~A. McIlraith and Kilian~Q. Weinberger, editors, {\em
  Proceedings of the Thirty-Second {AAAI} Conference on Artificial
  Intelligence, (AAAI-18), the 30th innovative Applications of Artificial
  Intelligence (IAAI-18), and the 8th {AAAI} Symposium on Educational Advances
  in Artificial Intelligence (EAAI-18), New Orleans, Louisiana, USA, February
  2-7, 2018}, pages 6598--6606. {AAAI} Press, 2018.
\newblock URL:
  \url{https://www.aaai.org/ocs/index.php/AAAI/AAAI18/paper/view/16952}.

\bibitem{HeuleB15b}
Marijn J.~H. Heule and Armin Biere.
\newblock Compositional propositional proofs.
\newblock In Martin Davis, Ansgar Fehnker, Annabelle McIver, and Andrei
  Voronkov, editors, {\em Logic for Programming, Artificial Intelligence, and
  Reasoning - 20th International Conference, {LPAR-20} 2015, Suva, Fiji,
  November 24-28, 2015, Proceedings}, volume 9450 of {\em Lecture Notes in
  Computer Science}, pages 444--459. Springer, 2015.
\newblock \href {https://doi.org/10.1007/978-3-662-48899-7\_31}
  {\path{doi:10.1007/978-3-662-48899-7\_31}}.

\bibitem{HeuleB18}
Marijn J.~H. Heule and Armin Biere.
\newblock What a difference a variable makes.
\newblock In Dirk Beyer and Marieke Huisman, editors, {\em Tools and Algorithms
  for the Construction and Analysis of Systems - 24th International Conference,
  {TACAS} 2018, Held as Part of the European Joint Conferences on Theory and
  Practice of Software, {ETAPS} 2018, Thessaloniki, Greece, April 14-20, 2018,
  Proceedings, Part {II}}, volume 10806 of {\em Lecture Notes in Computer
  Science}, pages 75--92. Springer, 2018.
\newblock \href {https://doi.org/10.1007/978-3-319-89963-3\_5}
  {\path{doi:10.1007/978-3-319-89963-3\_5}}.

\bibitem{HeuleKB17}
Marijn J.~H. Heule, Benjamin Kiesl, and Armin Biere.
\newblock Short proofs without new variables.
\newblock In Leonardo de~Moura, editor, {\em Automated Deduction - {CADE} 26 -
  26th International Conference on Automated Deduction, Gothenburg, Sweden,
  August 6-11, 2017, Proceedings}, volume 10395 of {\em Lecture Notes in
  Computer Science}, pages 130--147. Springer, 2017.
\newblock \href {https://doi.org/10.1007/978-3-319-63046-5\_9}
  {\path{doi:10.1007/978-3-319-63046-5\_9}}.

\bibitem{HeuleKM16}
Marijn J.~H. Heule, Oliver Kullmann, and Victor~W. Marek.
\newblock Solving and verifying the boolean pythagorean triples problem via
  cube-and-conquer.
\newblock In Nadia Creignou and Daniel~Le Berre, editors, {\em Theory and
  Applications of Satisfiability Testing - {SAT} 2016 - 19th International
  Conference, Bordeaux, France, July 5-8, 2016, Proceedings}, volume 9710 of
  {\em Lecture Notes in Computer Science}, pages 228--245. Springer, 2016.
\newblock \href {https://doi.org/10.1007/978-3-319-40970-2\_15}
  {\path{doi:10.1007/978-3-319-40970-2\_15}}.

\bibitem{JarvisaloBH10}
Matti J{\"{a}}rvisalo, Armin Biere, and Marijn Heule.
\newblock Blocked clause elimination.
\newblock In Javier Esparza and Rupak Majumdar, editors, {\em Tools and
  Algorithms for the Construction and Analysis of Systems, 16th International
  Conference, {TACAS} 2010, Held as Part of the Joint European Conferences on
  Theory and Practice of Software, {ETAPS} 2010, Paphos, Cyprus, March 20-28,
  2010. Proceedings}, volume 6015 of {\em Lecture Notes in Computer Science},
  pages 129--144. Springer, 2010.
\newblock \href {https://doi.org/10.1007/978-3-642-12002-2\_10}
  {\path{doi:10.1007/978-3-642-12002-2\_10}}.

\bibitem{JarvisaloBH12}
Matti J{\"{a}}rvisalo, Armin Biere, and Marijn Heule.
\newblock Simulating circuit-level simplifications on {CNF}.
\newblock {\em J. Autom. Reason.}, 49(4):583--619, 2012.
\newblock \href {https://doi.org/10.1007/s10817-011-9239-9}
  {\path{doi:10.1007/s10817-011-9239-9}}.

\bibitem{JarvisaloHB12}
Matti J{\"{a}}rvisalo, Marijn Heule, and Armin Biere.
\newblock Inprocessing rules.
\newblock In Bernhard Gramlich, Dale Miller, and Uli Sattler, editors, {\em
  Automated Reasoning - 6th International Joint Conference, {IJCAR} 2012,
  Manchester, UK, June 26-29, 2012. Proceedings}, volume 7364 of {\em Lecture
  Notes in Computer Science}, pages 355--370. Springer, 2012.
\newblock \href {https://doi.org/10.1007/978-3-642-31365-3\_28}
  {\path{doi:10.1007/978-3-642-31365-3\_28}}.

\bibitem{KieslRH18}
Benjamin Kiesl, Adri{\'{a}}n Rebola{-}Pardo, and Marijn J.~H. Heule.
\newblock Extended resolution simulates {DRAT}.
\newblock In Didier Galmiche, Stephan Schulz, and Roberto Sebastiani, editors,
  {\em Automated Reasoning - 9th International Joint Conference, {IJCAR} 2018,
  Held as Part of the Federated Logic Conference, FloC 2018, Oxford, UK, July
  14-17, 2018, Proceedings}, volume 10900 of {\em Lecture Notes in Computer
  Science}, pages 516--531. Springer, 2018.
\newblock \href {https://doi.org/10.1007/978-3-319-94205-6\_34}
  {\path{doi:10.1007/978-3-319-94205-6\_34}}.

\bibitem{KieslSTB18}
Benjamin Kiesl, Martina Seidl, Hans Tompits, and Armin Biere.
\newblock Local redundancy in {SAT:} generalizations of blocked clauses.
\newblock {\em Log. Methods Comput. Sci.}, 14(4), 2018.
\newblock \href {https://doi.org/10.23638/LMCS-14(4:3)2018}
  {\path{doi:10.23638/LMCS-14(4:3)2018}}.

\bibitem{Lammich17}
Peter Lammich.
\newblock Efficient verified {(UN)SAT} certificate checking.
\newblock In Leonardo de~Moura, editor, {\em Automated Deduction - {CADE} 26 -
  26th International Conference on Automated Deduction, Gothenburg, Sweden,
  August 6-11, 2017, Proceedings}, volume 10395 of {\em Lecture Notes in
  Computer Science}, pages 237--254. Springer, 2017.
\newblock \href {https://doi.org/10.1007/978-3-319-63046-5\_15}
  {\path{doi:10.1007/978-3-319-63046-5\_15}}.

\bibitem{MantheyHB12}
Norbert Manthey, Marijn Heule, and Armin Biere.
\newblock Automated reencoding of boolean formulas.
\newblock In Armin Biere, Amir Nahir, and Tanja E.~J. Vos, editors, {\em
  Hardware and Software: Verification and Testing - 8th International Haifa
  Verification Conference, {HVC} 2012, Haifa, Israel, November 6-8, 2012.
  Revised Selected Papers}, volume 7857 of {\em Lecture Notes in Computer
  Science}, pages 102--117. Springer, 2012.
\newblock \href {https://doi.org/10.1007/978-3-642-39611-3\_14}
  {\path{doi:10.1007/978-3-642-39611-3\_14}}.

\bibitem{MantheyP14}
Norbert Manthey and Tobias Philipp.
\newblock Formula simplifications as {DRAT} derivations.
\newblock In Carsten Lutz and Michael Thielscher, editors, {\em {KI} 2014:
  Advances in Artificial Intelligence - 37th Annual German Conference on AI,
  Stuttgart, Germany, September 22-26, 2014. Proceedings}, volume 8736 of {\em
  Lecture Notes in Computer Science}, pages 111--122. Springer, 2014.
\newblock \href {https://doi.org/10.1007/978-3-319-11206-0\_12}
  {\path{doi:10.1007/978-3-319-11206-0\_12}}.

\bibitem{NadelRS13}
Alexander Nadel, Vadim Ryvchin, and Ofer Strichman.
\newblock Efficient {MUS} extraction with resolution.
\newblock In {\em Formal Methods in Computer-Aided Design, {FMCAD} 2013,
  Portland, OR, USA, October 20-23, 2013}, pages 197--200. {IEEE}, 2013.
\newblock URL: \url{http://ieeexplore.ieee.org/document/6679410/}.

\bibitem{PhilippR16}
Tobias Philipp and Adri{\'{a}}n Rebola{-}Pardo.
\newblock {DRAT} proofs for {XOR} reasoning.
\newblock In Loizos Michael and Antonis~C. Kakas, editors, {\em Logics in
  Artificial Intelligence - 15th European Conference, {JELIA} 2016, Larnaca,
  Cyprus, November 9-11, 2016, Proceedings}, volume 10021 of {\em Lecture Notes
  in Computer Science}, pages 415--429, 2016.
\newblock \href {https://doi.org/10.1007/978-3-319-48758-8\_27}
  {\path{doi:10.1007/978-3-319-48758-8\_27}}.

\bibitem{PhilippR17}
Tobias Philipp and Adri{\'{a}}n Rebola{-}Pardo.
\newblock Towards a semantics of unsatisfiability proofs with inprocessing.
\newblock In Thomas Eiter and David Sands, editors, {\em LPAR-21, 21st
  International Conference on Logic for Programming, Artificial Intelligence
  and Reasoning, Maun, Botswana, May 7-12, 2017}, volume~46 of {\em EPiC Series
  in Computing}, pages 65--84. EasyChair, 2017.
\newblock URL: \url{https://easychair.org/publications/paper/V8G}.

\bibitem{PipatsrisawatD11}
Knot Pipatsrisawat and Adnan Darwiche.
\newblock On the power of clause-learning {SAT} solvers as resolution engines.
\newblock {\em Artif. Intell.}, 175(2):512--525, 2011.
\newblock \href {https://doi.org/10.1016/j.artint.2010.10.002}
  {\path{doi:10.1016/j.artint.2010.10.002}}.

\bibitem{Rebola-PardoC18}
Adri{\'{a}}n Rebola{-}Pardo and Lu{\'{\i}}s Cruz{-}Filipe.
\newblock Complete and efficient {DRAT} proof checking.
\newblock In Nikolaj Bj{\o}rner and Arie Gurfinkel, editors, {\em 2018 Formal
  Methods in Computer Aided Design, {FMCAD} 2018, Austin, TX, USA, October 30 -
  November 2, 2018}, pages 1--9. {IEEE}, 2018.
\newblock \href {https://doi.org/10.23919/FMCAD.2018.8602993}
  {\path{doi:10.23919/FMCAD.2018.8602993}}.

\bibitem{Rebola-PardoS18}
Adri{\'{a}}n Rebola{-}Pardo and Martin Suda.
\newblock A theory of satisfiability-preserving proofs in {SAT} solving.
\newblock In Gilles Barthe, Geoff Sutcliffe, and Margus Veanes, editors, {\em
  {LPAR-22.} 22nd International Conference on Logic for Programming, Artificial
  Intelligence and Reasoning, Awassa, Ethiopia, 16-21 November 2018}, volume~57
  of {\em EPiC Series in Computing}, pages 583--603. EasyChair, 2018.
\newblock URL: \url{https://easychair.org/publications/paper/zr7z}.

\bibitem{Rebola-PardoW20}
Adri{\'{a}}n Rebola{-}Pardo and Georg Weissenbacher.
\newblock {RAT} elimination.
\newblock In Elvira Albert and Laura Kov{\'{a}}cs, editors, {\em {LPAR} 2020:
  23rd International Conference on Logic for Programming, Artificial
  Intelligence and Reasoning, Alicante, Spain, May 22-27, 2020}, volume~73 of
  {\em EPiC Series in Computing}, pages 423--448. EasyChair, 2020.
\newblock URL: \url{https://easychair.org/publications/paper/cMtF}.

\bibitem{reckhow75_phd}
Robert~A. Reckhow.
\newblock {\em On the lengths of proofs in the propositional calculus}.
\newblock PhD thesis, University of Toronto, 1975.

\bibitem{SilvaS96}
Jo{\~{a}}o P.~Marques Silva and Karem~A. Sakallah.
\newblock {GRASP} - a new search algorithm for satisfiability.
\newblock In Rob~A. Rutenbar and Ralph H. J.~M. Otten, editors, {\em
  Proceedings of the 1996 {IEEE/ACM} International Conference on Computer-Aided
  Design, {ICCAD} 1996, San Jose, CA, USA, November 10-14, 1996}, pages
  220--227. {IEEE} Computer Society / {ACM}, 1996.
\newblock \href {https://doi.org/10.1109/ICCAD.1996.569607}
  {\path{doi:10.1109/ICCAD.1996.569607}}.

\bibitem{SinzB06}
Carsten Sinz and Armin Biere.
\newblock Extended resolution proofs for conjoining bdds.
\newblock In Dima Grigoriev, John Harrison, and Edward~A. Hirsch, editors, {\em
  Computer Science - Theory and Applications, First International Computer
  Science Symposium in Russia, {CSR} 2006, St. Petersburg, Russia, June 8-12,
  2006, Proceedings}, volume 3967 of {\em Lecture Notes in Computer Science},
  pages 600--611. Springer, 2006.
\newblock \href {https://doi.org/10.1007/11753728\_60}
  {\path{doi:10.1007/11753728\_60}}.

\bibitem{Soos10}
Mate Soos.
\newblock Enhanced gaussian elimination in dpll-based {SAT} solvers.
\newblock In Daniel~Le Berre, editor, {\em {POS-10.} Pragmatics of SAT,
  Edinburgh, UK, July 10, 2010}, volume~8 of {\em EPiC Series in Computing},
  pages 2--14. EasyChair, 2010.
\newblock URL: \url{https://easychair.org/publications/paper/j1D}.

\bibitem{SoosNC09}
Mate Soos, Karsten Nohl, and Claude Castelluccia.
\newblock Extending {SAT} solvers to cryptographic problems.
\newblock In Oliver Kullmann, editor, {\em Theory and Applications of
  Satisfiability Testing - {SAT} 2009, 12th International Conference, {SAT}
  2009, Swansea, UK, June 30 - July 3, 2009. Proceedings}, volume 5584 of {\em
  Lecture Notes in Computer Science}, pages 244--257. Springer, 2009.
\newblock \href {https://doi.org/10.1007/978-3-642-02777-2\_24}
  {\path{doi:10.1007/978-3-642-02777-2\_24}}.

\bibitem{TanHM21}
Yong~Kiam Tan, Marijn J.~H. Heule, and Magnus~O. Myreen.
\newblock cake{\_}lpr: Verified propagation redundancy checking in cakeml.
\newblock In Jan~Friso Groote and Kim~Guldstrand Larsen, editors, {\em Tools
  and Algorithms for the Construction and Analysis of Systems - 27th
  International Conference, {TACAS} 2021, Held as Part of the European Joint
  Conferences on Theory and Practice of Software, {ETAPS} 2021, Luxembourg
  City, Luxembourg, March 27 - April 1, 2021, Proceedings, Part {II}}, volume
  12652 of {\em Lecture Notes in Computer Science}, pages 223--241. Springer,
  2021.
\newblock \href {https://doi.org/10.1007/978-3-030-72013-1\_12}
  {\path{doi:10.1007/978-3-030-72013-1\_12}}.

\bibitem{Tseitin1983}
G.~S. Tseitin.
\newblock On the complexity of derivation in propositional calculus.
\newblock In J{\"o}rg~H. Siekmann and Graham Wrightson, editors, {\em
  Automation of Reasoning: 2: Classical Papers on Computational Logic
  1967--1970}, pages 466--483. Springer Berlin Heidelberg, 1983.
\newblock \href {https://doi.org/10.1007/978-3-642-81955-1_28}
  {\path{doi:10.1007/978-3-642-81955-1_28}}.

\bibitem{Urquhart87}
Alasdair Urquhart.
\newblock Hard examples for resolution.
\newblock {\em J. {ACM}}, 34(1):209--219, 1987.
\newblock \href {https://doi.org/10.1145/7531.8928}
  {\path{doi:10.1145/7531.8928}}.

\bibitem{Urquhart99}
Alasdair Urquhart.
\newblock The symmetry rule in propositional logic.
\newblock {\em Discret. Appl. Math.}, 96-97:177--193, 1999.
\newblock \href {https://doi.org/10.1016/S0166-218X(99)00039-6}
  {\path{doi:10.1016/S0166-218X(99)00039-6}}.

\bibitem{WetzlerHH14}
Nathan Wetzler, Marijn Heule, and Warren A.~Hunt Jr.
\newblock {DRAT}-trim: Efficient checking and trimming using expressive clausal
  proofs.
\newblock In Carsten Sinz and Uwe Egly, editors, {\em Theory and Applications
  of Satisfiability Testing - {SAT} 2014 - 17th International Conference, Held
  as Part of the Vienna Summer of Logic, {VSL} 2014, Vienna, Austria, July
  14-17, 2014. Proceedings}, volume 8561 of {\em Lecture Notes in Computer
  Science}, pages 422--429. Springer, 2014.
\newblock \href {https://doi.org/10.1007/978-3-319-09284-3\_31}
  {\path{doi:10.1007/978-3-319-09284-3\_31}}.

\bibitem{ZhangM03}
Lintao Zhang and Sharad Malik.
\newblock Validating {SAT} solvers using an independent resolution-based
  checker: Practical implementations and other applications.
\newblock In {\em 2003 Design, Automation and Test in Europe Conference and
  Exposition {(DATE} 2003), 3-7 March 2003, Munich, Germany}, pages
  10880--10885. {IEEE} Computer Society, 2003.
\newblock URL:
  \url{http://doi.ieeecomputersociety.org/10.1109/DATE.2003.10014}, \href
  {https://doi.org/10.1109/DATE.2003.10014}
  {\path{doi:10.1109/DATE.2003.10014}}.

\end{thebibliography}

\end{document}